\documentclass{article}
\usepackage{amsmath,inputenc}

\newif\ifarxiv   
\arxivtrue 

\ifarxiv
\else
\fi

\usepackage{graphicx}
\usepackage{csvsimple}
\usepackage{rotating}
\usepackage[letterpaper, portrait, margin=1in]{geometry}
\usepackage{hyperref}
\hypersetup{
    colorlinks=true,
    linkcolor=blue,
    citecolor=red,
    filecolor=magenta,      
    urlcolor=cyan
    }

\usepackage[nocomma]{optidef}
\usepackage{listings}
\usepackage{comment}
\usepackage{appendix}
\usepackage[ruled,vlined]{algorithm2e}
\usepackage{amsfonts} 
\usepackage{amssymb}
\usepackage{bm}
\usepackage[square,comma,numbers]{natbib}

\usepackage{color} 
\definecolor{mygreen}{RGB}{28,172,0} 
\definecolor{mylilas}{RGB}{170,55,241}
\DeclareFixedFont{\ttb}{T1}{txtt}{bx}{n}{12} 
\DeclareFixedFont{\ttm}{T1}{txtt}{m}{n}{12}  
\usepackage{bbm}
\usepackage{amsmath,amsthm}

\newtheorem{theorem}{Theorem}
\newtheorem{corollary}{Corollary}
\newtheorem{proposition}{Proposition}

\newtheorem{lemma}{Lemma}
\newtheorem{observation}{Observation}

\theoremstyle{definition}
\newtheorem{definition}{Definition}
\newtheorem{remark}{Remark}

\usepackage{color}
\definecolor{deepblue}{rgb}{0,0,0.5}
\definecolor{deepred}{rgb}{0.6,0,0}
\definecolor{deepgreen}{rgb}{0,0.5,0}

\usepackage{listings}
\usepackage{booktabs}

\renewcommand\footnotemark{}

\def\x{\bm{x}}
\def\f{\bm{f}}
\def\h{\bm{h}}
\def\r{\bm{r}}
\def\rr{\bm{r}}
\def\cc{\bm{c}}
\def\q{\bm{q}}
\def\c{\bm{c}}

\def\ddb{\bm{d}}
\def\P{\mathcal{P}}

\def\G{\mathcal{G}}
\def\0{\bm{0}}

\def\ttau{\bm{\tau}}

\def\qmin{q_{\text{min}}}
\def\Gmin{\G_{\text{min}}}
\def\qnext{q_{\text{next}}}

\DeclareMathOperator*{\argmin}{arg\,min}

\title{When Efficiency meets Equity in \\ Congestion Pricing and Revenue Refunding Schemes}

\author{Devansh Jalota$^1$, Kiril Solovey$^2$, Karthik Gopalakrishnan$^1$, Stephen Zoepf$^3$, \\ Hamsa Balakrishnan$^4$,  and Marco Pavone$^1$%
\thanks{$^1$ Stanford University, USA; {\tt \{djalota, kgopalakrishnan@stanford.edu, pavone\}@stanford.edu}. This work was supported in part by the National Science Foundation, CAREER Award CMMI1454737, and by the Stanford Institute for Human-Centered AI.}%
\thanks{$^2$ Technion - Israel Institute of Technology, Israel; {\tt kirilsol@technion.ac.il}.}%
\thanks{$^3$ Lacuna, Palo Alto, USA; {\tt stephen.zoepf@lacuna.ai}.}%
\thanks{$^4$ Massachusetts Institute of Technology, USA; {\tt hamsa@mit.edu}.}%
}

\date{March, 2023}

\begin{document}

\maketitle

\begin{abstract}
Congestion pricing has long been hailed as a means to mitigate traffic congestion; however, its practical adoption has been limited due to the resulting social inequity issue, e.g., low-income users are priced out off certain roads. This issue has spurred interest in the design of equitable mechanisms that aim to refund the collected toll revenues as lump-sum transfers to users. Although revenue refunding has been extensively studied for over three decades, there has been no thorough characterization of how such schemes can be designed to simultaneously achieve system efficiency and equity objectives. In this work, we bridge this gap through the study of \emph{congestion pricing and revenue refunding} (CPRR) schemes in non-atomic congestion games. We first develop CPRR schemes, which in comparison to the untolled case, simultaneously increase system efficiency without worsening wealth inequality, while being \emph{user-favorable}: irrespective of their initial wealth or values-of-time (which may differ across users), users would experience a lower travel cost after the implementation of the proposed scheme. We then characterize the set of optimal user-favorable CPRR schemes that simultaneously maximize system efficiency and minimize wealth inequality. Finally, we provide a concrete methodology for computing optimal CPRR schemes and also highlight additional equilibrium properties of these schemes under different models of user behavior. Overall, our work demonstrates that through appropriate refunding policies we can design user-favorable CPRR schemes that maximize system efficiency while reducing wealth inequality.
\end{abstract}

\section{Introduction}

The study of road congestion pricing is central to transportation economics and traces back to 1920 with the seminal work of Pigou \cite{pigou}. Since then, the marginal cost pricing of roads, where users pay for the externalities they impose on others, has been widely accepted as a mechanism to alleviate traffic congestion. In particular, congestion pricing can be used to steer users away from the user equilibrium (UE) traffic pattern~\cite{tsekeris2009design}, which forms when users selfishly minimize their own travel times\ifarxiv~\cite{how-bad-is-selfish,roughgarden2005selfish}\fi, towards the system optimum (SO)~\cite{Sheffi1985}. Despite the system-wide benefits of congestion pricing, its practical adoption has been limited~\cite{SMALL2001310}. A primary driving force behind the public opposition to congestion pricing has been the resultant inequity, e.g., high income users are likely to get the most benefit with shorter travel times while low income users suffer \ifarxiv exceedingly \fi large travel times since they avoid the high toll roads. Several empirical works have noted the regressive nature of congestion pricing\ifarxiv~\cite{eliasson2001road,manville-empirical} \fi and a recent theoretical work by \cite{gemici_et_al:LIPIcs:2019:10270} has also characterized the influence of road tolls on wealth inequality. In particular,\ifarxiv the latter paper\fi~\cite{gemici_et_al:LIPIcs:2019:10270} developed an \textit{Inequity Theorem} for users travelling between the same origin-destination (O-D) pair, and proved that any form of road tolls would increase the wealth inequality. These rigorous critiques are complemented by opinions expressed in the popular press that congestion fees amount to ``a tax on the working class~\cite{nyt-cp}.''

The lack of support for congestion pricing due to its social inequity issues\ifarxiv~\cite{CP-Low-acceptance,wachs2005then} \fi has led to a growing interest in the design of congestion-pricing schemes that are more equitable\ifarxiv~\cite{WU20121273,fair-deSouza}\fi. One approach that has been proposed to alleviate the inequity issues of congestion pricing is direct revenue redistribution, i.e., refunding the toll revenues to users in the form of lump-sum transfers. The idea of revenue refunding is analogous to that of \href{https://www.globalfueleconomy.org/transport/gfei/autotool/approaches/economic_instruments/fee_bate.asp}{feebates},
where refunds are used as a means to induce desirable behavior in society. Our work is centered on the design of congestion pricing and revenue refunding (CPRR) schemes that improve system performance without reducing wealth inequality, and benefit every user irrespective of their wealth or value-of-time. We view our work as paving the way for the design of practical, sustainable, and publicly acceptable congestion pricing schemes.

\ifarxiv
\paragraph{Contributions.}
\fi

In this work, we present the first study of the wealth-inequality effects of CPRR schemes in non-atomic congestion games, with a specific focus on devising CPRR schemes that simultaneously reduce the total system cost, i.e., the sum of the travel times on all edges of the network weighted by the corresponding values-of-time of users, without increasing the level of wealth (or income) inequality. We consider the setting of heterogeneous users, with differing values-of-time and income, who seek to minimize their individual travel cost, which is a linear function of their travel times, tolls, and refunds, in the system. As in previous work \cite{gemici_et_al:LIPIcs:2019:10270}, we incorporate the income elasticity of travel time, i.e., increased travel time corresponds to lost income, to reason about the income distribution of users before and after the imposition of a CPRR scheme.

\ifarxiv

To capture the behavior of selfish users, we study the effect of the Nash equilibria induced by CPRR schemes on the level of wealth inequality in society for non-atomic congestion games. We consider two notions of equilibrium formation: (i) exogenous equilibrium, wherein users minimize a linear function of their travel time and tolls, without considering refunds, as in~\cite{GUO2010972}, and (ii) endogenous equilibrium, a new notion we introduce, wherein coalitions of users additionally consider refunds in their travel cost minimization. For these two notions of equilibria, our contributions are four-fold:

\begin{enumerate}
    \item \emph{We develop CPRR schemes that improve both system efficiency and wealth inequality, while being favorable to all users.} Under the exogenous equilibrium model, we establish the existence of a CPRR scheme that, compared with the untolled outcome, (i) is user-favorable, i.e., every user group, irrespective of their initial wealth, has a lower travel cost after the implementation of the scheme, (ii) lowers total system cost, and (iii) does not increase wealth inequality (see \ifarxiv Figure~\ref{fig:combined}\else Fig.~\ref{fig:combined}\fi). We call such CPPR schemes \emph{Pareto improving}. For the case when all travel demand is between a single O-D pair and each user's value-of-time is proportional to their income, we further show that the same CPRR scheme reduces wealth inequality relative to the ex-ante income distribution, i.e., the income profiles of users prior to making their trips. Thus, our results show that it is possible to reverse the wealth-inequality effects of congestion pricing that were established in the \textit{Inequity Theorem} in~\cite{gemici_et_al:LIPIcs:2019:10270} through appropriate revenue refunding schemes.
    
    \item \emph{We characterize the set of optimal CPRR schemes that are favorable to all users.} 
    In particular, we establish in the exogenous equilibrium setting that there exist CPRR schemes that simultaneously minimize total system cost and level of wealth inequality among all CPRR schemes that are favorable to any user (see \ifarxiv Figure~\ref{fig:combined}\else Fig.~\ref{fig:combined}\fi). To establish this claim, we (i) characterize the income distribution of users after the completion of their trip, and (ii) prove a monotonic relationship between the minimum achievable level of wealth-inequality and the total system cost. 
    \item \emph{We develop a method to compute the optimal CPRR scheme, and, in particular, for a commonly used wealth inequality measure, the discrete Gini coefficient, we show that a simple max-min allocation of the refunds among user groups with different levels of income will be optimal.}
    \item \emph{Finally, we highlight additional equilibrium properties of such schemes when users endogenize the effect of refunds on their travel decisions.} This model of user behavior wherein users minimize a linear function of not only their travel times and tolls but also their refunds is natural as users may additionally account for their received refunds when making travel decisions. In this setting, we show that the optimal CPRR scheme is robust to coalitions, i.e., any exogenous equilibrium induced by an optimal CPRR scheme is also an endogenous equilibrium with coalitions.
\end{enumerate}

\fi

We remark that in line with prior literature on traffic routing with heterogeneous groups of users~\cite{gemici_et_al:LIPIcs:2019:10270,heterogeneous-pricing-roughgarden,multicommodity-extension}, in this work, we assume a complete information setting wherein the different attributes (i.e., the income, value-of-time, and O-D pair) of the user groups are known and can be used in the design of CPRR schemes. To this end, our results can be interpreted as the theoretical limits of what is achievable in terms of the efficiency and equity outcomes given perfect state information in the traffic routing context. However, we remark that even though we consider the complete information setting wherein the tolls and refunds are computed in a centralized manner, the developed optimal CPRR schemes induce selfish users to distributedly optimize their individual objectives and collectively enforce a traffic pattern
that minimizes both total system cost and the level of wealth inequality. Furthermore, since our results contribute to the vast literature on designing intervention and control schemes under perfect state information~\cite{pan1993h, sofronova2020traffic, sabag2021regret}, our proposed approach to designing CPRR schemes serve as more of a design module rather than an end-to-end solution to the equity problem associated with congestion pricing. We do note, however, that there are several methods to estimate user attributes, e.g., their values of time or preferences, that have been explored in the empirical literature~\cite{arora2020private,Buchholzrdab050,cohen2016using}, which can be used to inform the inputs that are necessary for the design of optimal CPRR schemes. 

\begin{figure*}[!ht]
      \centering
      \includegraphics[width=0.75\linewidth]{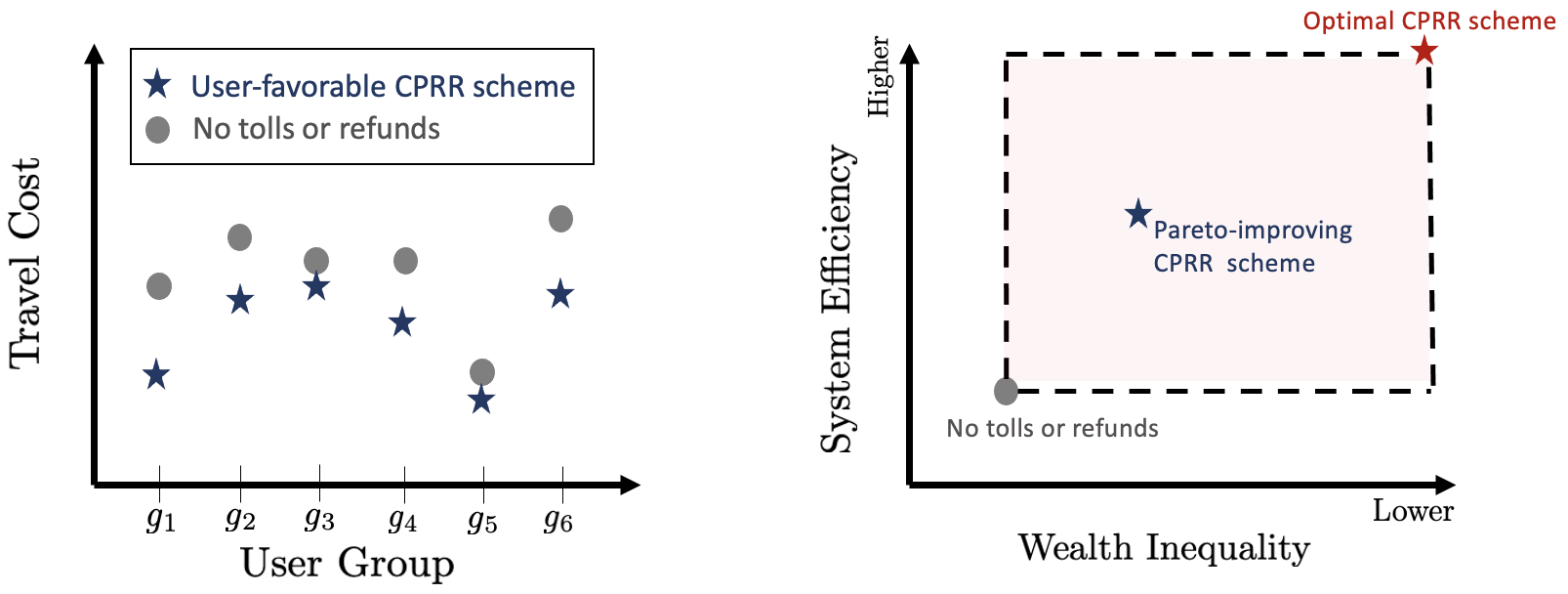}
      \caption{{\small \sf \ifarxiv We develop congestion pricing with revenue refunding (CPRR) schemes that improve both system efficiency and wealth inequality, while being \emph{favorable} to all users. On the left we illustrate the concept of a user-favorable CPRR scheme, where every user group, irrespective of their initial income, has a lower travel cost, which includes the cost of travel time, tolls and refunds, after the implementation of the scheme compared with the untolled setting. Note that users in a given group have the same value-of-time, income, and O-D pair. On the right we illustrate the concept of \emph{Pareto-improving} user-favorable CPRR schemes, which increase system efficiency, and decrease wealth inequality, compared with the untolled case (the rectangular region, whose opposite corners are the optimal solution and no tolls or refunds solution, represents the set of all Pareto-improving outcomes of user-favorable schemes). We also illustrate an \emph{optimal} user-favorable CPRR scheme that simultaneously achieves the highest system efficiency and the lowest wealth inequality and establish the existence of such a scheme in this work. \fi
      \label{fig:combined}}} 
   \end{figure*}

Overall, our work demonstrates that if we utilize the collected toll revenues to devise appropriate refunding policies then we can achieve system efficiency without increasing inequality. Further, in doing so, we ensure that our designed schemes are publicly acceptable since we guarantee that each user is at least as well off as before the introduction of the CPRR scheme. As a result, we view our work as a significant step in shifting the discussion around congestion pricing from one that has focused on the societal inequity impacts of road tolls to one that centers around \emph{how} to best preserve equity through the distribution of toll revenues. 

\ifarxiv
A preliminary version of this work appeared in~\cite{jalota-acm-eaamo}, and our work builds on that paper by developing a method to compute optimal CPRR schemes and analyzing additional equilibrium properties of these schemes. In particular, we present (i) methodological extensions, including an example to show how the optimal revenue refunding scheme can be derived for the discrete Gini coefficient wealth inequality measure, and (ii) a study of the wealth inequality effects of CPRR schemes under an endogenous equilibrium setting, where users also take refunds into account in their travel cost minimization.
\fi

\ifarxiv
\paragraph{Organization.}
This paper is organized as follows. Section~\ref{sec:related-lit} reviews related literature. We then present a model of traffic flow as well as metrics to evaluate the inequality of the wealth distribution and the efficiency of a traffic assignment in Section~\ref{sec:model}. We prove the existence of Pareto improving and optimal CPRR schemes for the exogenous setting in Sections~\ref{sec:existence}, and~\ref{sec:optimality}, respectively. Next, we present a method to compute optimal CPRR schemes and show that any exogenous equilibrium of an optimal CPRR scheme is also an endogenous equilibrium with coalitions in Section~\ref{sec:opt-CPRR}. Finally, we discuss how our work fits into the broader conversation around equitable transportation in Section~\ref{sec:discussion}. We conclude the paper and provide directions for future work in Section~\ref{sec:conclusion}.
\fi

\section{Related Work} \label{sec:related-lit}

The design of mechanisms that satisfy both system efficiency and user fairness desiderata has been a centerpiece of algorithm design for a range of applications. For instance, in resource allocation settings, ~\cite{Bertsimas-price-of-fairness} quantified the loss in efficiency 
when the allocation outcomes are required to satisfy certain fairness criteria. In machine learning classification tasks,~\cite{dwork-fairness-through-awareness} studied group-based fairness notions to prevent discrimination against individuals belonging to disadvantaged groups. In the context of traffic routing,~\cite{so-routing-seminal} introduced a fairness-constrained traffic-assignment problem to achieve a balance between the total travel time of a traffic assignment and its level of fairness. Here, fairness is measured through the maximum ratio between the travel times of users travelling between the same O-D pair\ifarxiv~\cite{Roughgarden2002HowUI}. \fi \ifarxiv Subsequent work on fair traffic routing has focused on developing algorithms to solve the fairness constrained traffic assignment problem\ifarxiv~\cite{ANGELELLI20161,ANGELELLI2018234,ANGELELLI2020} \else~\cite{ANGELELLI20161} \fi and road tolling schemes to enforce the fairness constrained flows in practice~\cite{Jalota.ea.CDC21.extended}. \fi

Resolving the efficiency and equity trade-off is particularly important for allocation mechanisms involving monetary transfers given \ifarxiv the potential negative welfare impacts of such mechanisms on low-income groups. \fi Although achieving system efficiency involves allocating goods to users with the highest willingness to pay, in many settings, e.g., cancer treatment, the needs of users are not well expressed by their willingness to pay~\cite{weitzman-seminal}. Since Weitzman's seminal work on accounting for agent's needs in allocation decisions~\cite{weitzman-seminal}, there has been a rich line of work on taking into account redistributive considerationsin resource allocation problems. For instance,~\cite{PPP-Besley-coate} analyzed the free provision of a low-quality public good to low-income users by taxing individuals that consume the same good of a higher quality in the private market. More recently, \cite{CONDORELLI2013582} studied the allocation of \ifarxiv identical \fi objects to agents with the objective of maximizing agent's values that may be different from their willingness to pay. \ifarxiv This analysis was then extended to the allocation of heterogeneous objects to a continuum of agents by \cite{RAM-Akbarpour}. \fi

In the context of congestion pricing, revenue redistribution has long been considered as a means to alleviate the inequity issues of congestion pricing~\cite{small1992using}. Several revenue redistribution strategies have been proposed in the literature, such as the lump-sum transfer of toll revenues to users~\cite{goodwin1989rule}. In the setting of Vickrey's bottleneck congestion model~\cite{vickrey1969congestion} (a benchmark representation of peak-period traffic congestion on a single lane),~\cite{arnott1994} investigated how a uniform lump-sum payment of toll revenues can be used to make heterogeneous users better off than prior to the implementation of the tolls and refunds. \ifarxiv \cite{DAGANZO1995139} developed a novel strategy in the bottleneck congestion model, wherein only a fraction of the users are tolled while the remaining users are exempt from tolls. To extend the application of revenue redistribution schemes to a two parallel-routes setting,~\cite{ADLER2001447} designed a mechanism wherein the revenue collected from users on the more desirable route was directly transferred to users travelling on the less desirable route. \fi In more general networks with a single O-D pair,~\cite{eliasson2001road} established the existence of a tolling mechanism with uniform revenue refunds that reduced the travel cost for each user while decreasing the total system travel time as compared to before the tolling reform. The extension of this result to general road networks with a multiple O-D pair travel demand and heterogeneous users was investigated by~\cite{GUO2010972}. While~\cite{GUO2010972} characterize conditions for the CPRR scheme to be user-favorable, our work studies the influence of such schemes by characterizing their influence on wealth inequality. \ifarxiv In particular, we design CPRR schemes that simultaneously reduce total system cost and do not increase the level of wealth inequality relative to that under the untolled user equilibrium outcome. \fi

\section{Preliminaries} \label{sec:model}
In this section, we introduce basic definitions and concepts regarding traffic flow, congestion pricing and revenue refunding (CPRR) schemes, and metrics for system efficiency and wealth-inequality.

\subsection{Elements of Traffic Flow}
We model the road network as a directed graph $G = (V, E)$, with the vertex and edge sets denoted by $V$ and $E$, respectively. Each edge $e \in E$ has a flow-dependent travel-time function $t_e: \mathbb{R}_{\geq 0} \rightarrow \mathbb{R}_{\geq 0}$, which maps $x_e$, the traffic flow rate on edge $e$, to the travel time $t_e(x_e)$. The flow rate $x_e$ on edge $e$ represents the average number of vehicles traversing through that edge during a fixed time interval (e.g., over an hour). As is standard in the literature, we assume that the function $t_e$, for each $e \in E$, is differentiable, convex and monotonically increasing. While our model of the edge travel time functions assumes that the edge travel times have infinite capacities, as is common in the non-atomic congestion game and transportation literature~\cite{gemici_et_al:LIPIcs:2019:10270,YANG20041}, we note that our model can be extended to the setting with hard capacity constraints for appropriate choices of the travel time functions that grow very steeply once the road capacities have been exceeded.

Users make trips in the road network and belong to a discrete set of user groups based on their (i) value-of-time, (ii) income, and (iii) O-D pair. Let $\G$ denote the set of all user groups, and let $v_g>0$, $q_g>0$ and $w_g = (s_g, u_g)$ denote the value-of-time, income, and O-D pair represented by an origin $s_g$ and destination $u_g$, respectively, for each user in group $g \in \G$. Each user belonging to a group $g$ makes a trip on a path, which is a sequence of  directed edges beginning at $s_g$ and ending at $u_g$ (without visiting any node more than once). The set of all possible paths between OD-pair $w_g$ is denoted as $\P_{g}$ and the travel demand $d_g$ of user group $g$ represents the total flow to be routed through paths in $\P_{g}$. 

A path flow pattern $\f = \{f_{P, g}: g \in \G, P \in \P_{g} \}$ specifies for each user group $g$, the amount of flow $f_{P,g}\geq 0$ routed on a path $P \in \P_{g}$. In particular, a flow $\f$ must satisfy the user demand, i.e., $\sum_{P \in \P_{g}} f_{P,g} = d_g, \text{ for all } g \in \G$. We denote the set of all non-negative flows that satisfy this constraint as~$\Omega$.

Each path flow $\f = \{f_{P, g}: g \in \G, P \in \P_{g} \}$ is associated with a corresponding edge flow $\x= \{x_e \}_{e \in E}$ and group specific edge flows $\x^g = \{x_e^g \}_{e \in E}$ for all $g \in \G$, where $x_e^{g}$ represents the flow of users in group $g$ on edge $e$. The relationship between the path and edge flows is given by
\ifarxiv
\begin{align}
    &\sum_{P \in \mathcal{P}_{g}: e \in P} f_{P,g} = x_e^{g}, \text{ for all } e \in E, \text{ and } g \in \G \label{eq:edge-constraint-by-class} \\
    &\sum_{g \in \mathcal{G}} x_e^{g} = x_e, \text{ for all } e \in E, \label{eq:edge-constraint}
\end{align}
where $P \in \mathcal{P}_{g}: e \in P$ denotes the set of all paths $P \in \P_g$ that include the edge $e$. 
\fi

\subsection{CPRR Schemes}

A congestion pricing and revenue refunding (CPRR) scheme is defined by a tuple $(\ttau, \rr)$, where (i) $\ttau = \{\tau_e: e \in E\}$ is a vector of edge prices (or tolls), and (ii) $\rr = \{r_g: g \in \G \}$ is a vector of group-specific revenue refunds, where each user in group $g$ receives a lump-sum transfer of $r_g$. In other words, everybody pays the same toll for using an edge independent of their group, and all users with the same income, value-of-time and O-D pair get the same refund, irrespective of the actual path they take between the O-D pair $w_g$. We note that the vector of refunds $\r$, in general, need not be non-negative and can take on any real values. Under the CPRR scheme $(\ttau, \rr)$ and a vector of edge flows $\x$, the total value of tolls collected is given by $\Pi \coloneqq \sum_{e \in E} \tau_e x_e$. In this work we consider CPPR schemes such that the sum of the revenue refunds equals the sum of the revenue collected from the edge tolls, i.e., $\sum_{g \in \G}  r_g d_{g} = \Pi$. In addition, we consider revenue refunding schemes that depend only on the groups $\G$ and the total revenue $\Pi$ induced by a flow $\f$, but not on the specific paths taken by users under $\f$. We leave the study of more complex path-dependent refunding schemes for future work (see Section~\ref{sec:conclusion}).

\ifarxiv
The total travel cost incurred by the user consists of two components. The first one is a linear function of their travel time and tolls, which is a commonly-used modelling approach~\cite{heterogeneous-pricing-roughgarden,multicommodity-extension}. In addition, we have a component which reflects the refund received. The overall model that we use here, which is formally defined below, has been previously considered in the literature~\cite{GUO2010972}.  
\fi

\begin{definition} [User Travel Cost] \label{def:user-travel-cost}
Consider a CPRR scheme $(\ttau, \rr)$ and a flow pattern $\f$ with edge flow $\x$. Then, the total cost incurred by a user belonging to a group $g \in \G$ when traversing a path $P \in \P_{g}$ with $f_{P,g}>0$ is given by 
\ifarxiv
\begin{align}
    \mu_P^g(\f, \ttau, \rr) \coloneqq \sum_{e \in P} \left( v_g t_e(x_e) + \tau_e \right) - r_g.
\end{align}
\fi
\end{definition}
With slight abuse of notation, we will denote $\mu_P^g(\f, \ttau, \0)$ as a travel cost that does not include refunds, and $\mu_P^g(\f, \0, \0)$ as a travel cost that does not account for tolls or refunds, where $\0$ is a vector of zeros. 

\subsection{System Efficiency and Wealth Inequality Metrics} \label{sec:metrics}
We evaluate the quality of a CPRR scheme using two metrics: (i) system efficiency, which is measured through the total system cost, and (ii) wealth inequality.

\ifarxiv
\paragraph{Total System Cost:} 
\fi
For any feasible path flow $\f$ with corresponding edge flows $\x$ and group specific edge flows $\x^g$, the total system cost $C(\f)$, is the sum of travel times weighted by the users' values-of-time  across all edges \ifarxiv of the network \cite{GUO2010972,heterogeneous-pricing-roughgarden,multicommodity-extension}, \fi i.e., \ifarxiv \[C(\f) := \sum_{e \in E} \sum_{g \in \G} v_g x_e^{g} t_e(x_e).\] \fi We denote by $C^*:=\min_{\f\in \Omega}C(\f)$ the widely studied cost-based system optimum.

\ifarxiv
\paragraph{Wealth Inequality:} 
\fi
We measure the impact of a CPPR scheme on wealth inequality in the following manner. For a profile of incomes $\q = \{q_g: g \in \G \}$, we let a function $W: \mathbb{R}^{|\G|}_{\geq 0} \rightarrow \mathbb{R}_{\geq 0}$ measure the level of wealth inequality of society. We say that an income distribution $\Tilde{\q}$ has a lower level of wealth inequality than $\q$ if and only if $W(\Tilde{\q}) \leq W(\q)$. 

\ifarxiv
In this work, we assume that the wealth-inequality measure $W(\cdot)$ satisfies the following properties:

\begin{enumerate}
    \item Scale Independence: The wealth-inequality measure remains unchanged after rescaling incomes by the same positive constant, i.e., $W(\lambda \q) = W(\q)$ for any $\lambda > 0$.
    \item Regressive Taxes Increase Inequality: The wealth-inequality measure increases if the incomes of users are scaled by constants that increase as the income increases. That is, for two income profiles $\q$ and $\Tilde{\q}$ with $\Tilde{q}_g = \delta_g q_g$, where $0<\delta_g \leq \delta_{g'}$ if $q_g \leq q_{g'}$ for any two groups $g, g'$, then $W(\Tilde{\q}) \geq W(\q)$.
    \item Progressive Taxes Decrease Inequality: The wealth-inequality measure decreases if the incomes of users are scaled by constants that decrease as the income increases. That is, for two income profiles $\q$ and $\Tilde{\q}$ with $\Tilde{q}_g = \delta_g q_g$, where $0<\delta_g \leq \delta_{g'}$ if $q_{g'} \leq q_{g}$ for any two groups $g, g'$, then $W(\Tilde{\q}) \leq W(\q)$.
\end{enumerate}
The above properties are well defined for any wealth inequality distribution when the incomes of all users are strictly positive, which we assume in this work. We note that the above properties are fairly natural \cite{gemici_et_al:LIPIcs:2019:10270,dabla2015causes} and hold for commonly used wealth-inequality measures, such as the discrete Gini coefficient, which we elucidate in detail in Section~\ref{sec:optimality}. Furthermore, we note that the above properties jointly imply the following important property of the wealth-inequality measure $W$:

\noindent \textbf{Constant Income Transfer Property}: If the initial income distribution is $\q$ and each user is transferred a non-negative (non-positive) amount of money $\lambda$ ($-\lambda$) where $0 \leq \lambda < \min_{g \in \G} q_{g}$, then the wealth inequality cannot increase (decrease). That is, $W(\q + \lambda \mathbf{1}) \leq W(\q)$ and $W(\q - \lambda \mathbf{1}) \geq W(\q)$, where $\mathbf{1}$ is a vector of ones. 

We defer a proof of how the constant income transfer property follows from the regressive and progressive tax properties to \ifarxiv Appendix~\ref{sec:const-inc-transfer}\else the extended version of our work~\cite{jalota-cprr}\fi.

\fi

When using the wealth inequality measure $W$, we are interested in understanding the influence of a flow~$\f$ for a given CPRR scheme $(\ttau, \rr)$ on the income distribution of users. To this end, we define the income profile of users before making their trip as the \emph{ex-ante income distribution} $\q^0>\mathbf{0}$ and that after making their trip as the \emph{ex-post income distribution}, which is defined as follows.

\begin{definition} [Ex-Post Income Distribution] \label{def:ex-post}
For a given CPRR scheme $(\ttau, \rr)$ and an equilibrium flow $\f$, the induced ex-post income distribution of users is denoted by $\q(\f, \ttau, \rr)$ and is defined as follows. For a given group $g$, we have that \ifarxiv \[q_g(\f, \ttau, \rr):=q_g^0-\beta\mu^g(\f,\ttau,\rr),\] \fi 
where $\q^0$ is the ex-ante income distribution and $\beta$ is a small constant such that the ex-post income of users is strictly positive and represents the relative importance of the congestion game under consideration to an individual's well-being~\cite{gemici_et_al:LIPIcs:2019:10270}. 
\end{definition}
We reiterate that the small constant $\beta$ does not depend on the type of trip being made or the importance of that trip to the user but solely reflects the importance of the congestion game under consideration to an individual’s well-being, as in~\cite{gemici_et_al:LIPIcs:2019:10270}. The positive income assumption ensures that the above defined wealth inequality properties (including scale independence) hold, which would not be the case if users have negative incomes.

We note that in this paper we consider time-invariant travel demand that is fixed for all user groups and assume fractional flows, both of which are standard assumptions in the \ifarxiv traffic routing~\cite{Patriksson15} and non-atomic congestion games~\cite{Roughgarden2002HowUI} literature. \fi In line with prior work by \cite{GUO2010972}, we assume that users are refunded based on their income, value-of-time, and O-D pair. Furthermore, similar to much of the prior literature in traffic routing with heterogeneous users~\cite{gemici_et_al:LIPIcs:2019:10270,heterogeneous-pricing-roughgarden,multicommodity-extension}, we assume that the different attributes (i.e., the income, value-of-time, and O-D pair) of the user groups are known, and can be used in the design of CPRR schemes. In practice, such centralized information on user attributes may not be known and we defer the problem of dealing with incomplete information settings to \ifarxiv future work (see Section~\ref{sec:discussion}). \fi

\section{Pareto Improving CPRR Schemes} \label{sec:existence}

The social inequity issue surrounding the regressive nature of congestion pricing has been documented in several empirical and theoretical works, while also having spurred political opposition to its implementation in practice. In this section, we show that if the tolls collected from congestion pricing are refunded to users in an appropriate way then the wealth inequality effects of congestion pricing can be reversed. Throughout this section and the next we assume 
that user behavior is characterized through the \emph{exogenous equilibrium} model wherein users minimize a linear function of their travel time and tolls, without considering refunds.

After formally defining exogenous equilibrium below, we develop a CPRR scheme that simultaneously decreases the total system cost of all users while not increasing the level of wealth inequality relative to the untolled outcome, a property which we refer to as \emph{Pareto improving}. Moreover, when designing the scheme, we ensure that it is politically acceptable \ifarxiv for implementation \fi by guaranteeing that each user is at least as well off in terms of the travel cost $\mu^g$, which includes travel time, tolls, and refunds, under the CPRR scheme than that without the implementation of congestion pricing or refunds.

Next, we consider the important special case \ifarxiv of travel demand \fi when users travel between the same O-D pair, and have values-of-time proportional to their income. In this setting, we establish the existence of a Pareto improving CPRR scheme that results in an ex-post income distribution that has a lower wealth inequality as compared to that of the ex-ante income distribution. Note that this result is stronger than the more general case with \ifarxiv multiple O-D pairs considered above, \fi as the wealth-inequality measure of the ex-ante income distribution is lower than that of the ex-post income distribution for the untolled case.

\subsection{Exogenous Equilibrium}
To capture the strategic behavior of users, we present below the standard model of Nash equilibrium with heterogeneous users, which we call exogenous equilibrium. 
The exogenous setting is commonly studied in the context of non-atomic congestion games without~\cite{heterogeneous-pricing-roughgarden,multicommodity-extension} or with refunds~\cite{GUO2010972}. As the name suggests, in an exogenous equilibrium the revenue refunds are assumed to be \emph{exogenous} and do not influence the behavior and route choice of users in the transportation network. That is,
users minimize a linear function of their travel time and tolls, without considering refunds. 

We note that such a model of user behavior can be quite realistic in certain settings, especially since accounting for refunds when making route choices may often be too complex and involve quite sophisticated decision making on the part of users. Furthermore, for users to reason about how their path choice will influence their refund, they must know the refunding policy, which may typically not be known in practice, thereby making the notion of an exogenous equilibrium more appropriate in such settings. We do consider the more sophisticated \emph{endogenous} setting in Section~\ref{sec:endogenous} and demonstrate that our results obtained in the exogenous setting also extend to endogenous setting as well.

The following definition formalizes the notion of an exogenous equilibrium, which only depends on the \ifarxiv congestion pricing \fi component $\ttau$ of a CPRR scheme $(\ttau,\rr)$. 

\begin{definition} [Exogenous Equilibrium] \label{def:exo-eq}
For a given congestion-pricing scheme  $\ttau$, a path flow pattern $\f$ is an exogenous equilibrium if for each group $g\in \G$ it holds that $f_{P,g}>0$ for some path $P \in \P_{g}$ if and only if
\begin{align*}
    \mu_{P}^g(\f, \ttau,\0) \leq \mu_{Q}^g(\f, \ttau, \0), \text{ for all } Q \in \P_{g}.
\end{align*}
\ifarxiv
In such a case, we say that $\f$ is an exogenous $\ttau$-equilibrium.
\fi
\end{definition}

We reiterate that the above notion of an exogenous equilibrium is the standard Nash equilibrium concept used in non-atomic congestion games and follows since users are infinitesimal, unlike equilibrium concepts in atomic congestion games or in the presence of coalitions (see Definition~\ref{def:endo-eq2}). In this work, we refer to this equilibrium concept as \emph{exogenous} to explicitly distinguish it from the \emph{endogenous} setting when users also account for refunds when making travel decisions. A key property of any exogenous $\ttau$-equilibrium $\f$ is that all users within a given group $g\in \G$ incur the same travel cost without refunds, irrespective of the path on which they travel. Hence, we drop the path dependence in the notation and denote the user travel cost without refunds for any user in group $g$ at flow $\f$ as $\mu^g(\f, \ttau, \0)$. Additionally, since the refund $r_g$ is the same for all users in group $g$, the travel cost with refunds is denoted as $\mu^g(\f, \ttau, \rr)$.

Another useful property of an exogenous equilibrium is that for a given congestion-pricing scheme $\ttau$, the resulting total system cost, user travel cost, and ex-post income distribution are invariant under the different $\ttau$-equilibria (see \ifarxiv Appendix~\ref{apdx:kkt-multiclass-ue} \fi for a discussion). That is for any two $\ttau$-equilibria $\f$ and $\f'$ it holds that $C(\f)=C(\f')$, $\mu^g(\f,\ttau,\0)=\mu^g(\f',\ttau,\0)$, and $\q(\f,\ttau,\rr)=\q(\f',\ttau,\rr)$. Thus, we will use the simplified notation $C_{\ttau}:=C(\f)$, $\mu^g(\ttau,\rr):=\mu^g(\f,\ttau,\rr)$, and $\q(\ttau,\rr):=\q(\f,\ttau,\rr)$ for any exogenous $\ttau$-equilibrium $\f$, when considering the exogenous equilibrium model. In this context, note that $C_{\0}$ corresponds to the untolled total system cost, and this quantity is identical for both the exogenous and endogenous equilibrium (we consider the latter in Section~\ref{sec:endogenous}).   

\subsection{User-Favorable Pareto Improving CPRR Schemes} \label{subsec:ufpiCPRR}
To ensure that the CPRR schemes we develop are politically acceptable, we consider schemes, as in\ifarxiv~\cite{GUO2010972,lawphongpanich2010solving}\fi, that result in equilibrium outcomes wherein each individual user is at least as well off as compared to that under the untolled user equilibrium outcome, a property we refer to as user-favorable (see \ifarxiv Figure~\ref{fig:combined}\fi). We note that the definition below readily extends to the setting of endogenous equilibria as well.

\begin{definition} [User-Favorable CPRR Schemes] \label{def:pareto-imp}
A CPRR scheme $(\ttau, \rr)$ is user-favorable if for any (exogenous) $\ttau$-equilibrium the travel cost of any user group $g$ does not increase with respect to any untolled $\0$-equilibrium $\f^0$, i.e., $\mu^g(\ttau, \rr) \leq \mu^g(\0,\0)$.
\end{definition}

We mention that the the above definition can readily be extended to incorporate the notion of a user-favorable CPRR scheme relative to any status-quo traffic equilibrium pattern, which is not necessarily equal to the untolled case, e.g., the traffic pattern in a city that has already implemented some form of congestion pricing. Thus, considering the untolled user equilibrium $\f^0$ in the above definition is without loss of generality. 

We now present the main result of this section. In particular, we establish that any pricing scheme $\ttau$ that improves the system efficiency compared to the untolled case can be paired with a revenue refunding scheme $\rr$ such that the wealth inequality relative to the ex-post income distribution under the untolled setting is not increased, i.e., the CPRR scheme $(\ttau,\rr)$ is Pareto improving (see \ifarxiv Figure~\ref{fig:combined}\fi) and user-favorable.  
Note that designing CPRR schemes that achieve a lower level of wealth inequality and total system cost as compared to that of the untolled user equilibrium outcome is desirable since this implies that the CPRR scheme improves upon both the system efficiency and fairness metrics relative to the status-quo traffic equilibrium pattern. \ifarxiv We discuss other useful aspects of this result in Section~\ref{sec:discussion}. \fi

\begin{proposition} [Existence of Pareto Improving CPRR Scheme] \label{prop:rev-refund-decreases-ineq}
Let $\ttau$ be a congestion-pricing scheme such that $C_{\ttau} \leq C_{\0}$, where $C_{\0}$ is the untolled total system cost. Then, there exists a refund scheme $\rr$ such that $(\ttau, \rr)$ is user-favorable and does not increase wealth inequality, i.e., $W(\q(\ttau, \rr)) \leq W(\q(\0,\0))$. That is, the scheme $(\ttau, \rr)$ is \emph{Pareto improving}.
\end{proposition}

\ifarxiv

Note that Proposition~\ref{prop:rev-refund-decreases-ineq} relies on the key observation that an exogenous equilibrium is completely defined through the road tolls $\ttau$, and is thus oblivious of the refund $\rr$. Furthermore, we remark that both Definition~\ref{def:pareto-imp} and Proposition~\ref{prop:rev-refund-decreases-ineq} can readily be extended to incorporate the notions of user-favorable and Pareto improving CPRR schemes relative to any status-quo traffic equilibrium pattern beyond the untolled user equilibrium. For simplicity, we prove those properties relative to the untolled setting. 
We now prove Proposition~\ref{prop:rev-refund-decreases-ineq} by leveraging a class of user-favorable CPRR schemes that were developed recently~\cite[Theorem 1]{GUO2010972}.

\begin{lemma} [Existence of user-favorable CPRR Scheme~\cite{GUO2010972}] \label{lem:PI-CPRR}
Let $\ttau$ be a congestion pricing scheme such that $C_{\ttau} \leq C_{\0}$. Then, for any $\alpha_g \geq 0$ with $\sum_{g \in \G} \alpha_g = 1$, the CPRR scheme $(\ttau,\rr)$ with refunds given by $r_g = \mu^g(\ttau,\0) - \mu^g(\0,\0) + \frac{\alpha_g}{d_g}(C_{\0} - C_{\ttau})$, for each group $g$ is user-favorable.
\end{lemma}

The above lemma states that as long as the edge tolls $\ttau$ reduce the total system cost, there exists a method to refund revenues that makes every user at least as well off as compared to that under the untolled case. Since this lemma is a slight modification to that shown in~\cite{GUO2010972}, (we consider weak inequalities $\alpha_g \geq 0$ and $C_{\ttau} \leq C_{\0}$, whereas~\cite{GUO2010972} considered strict inequalities $\alpha_g > 0$ and $C_{\ttau} < C_{\0}$) we present its proof in \ifarxiv Appendix~\ref{apdx:lemma-pi-cprr-pf} \else the extended version of this work~\cite{jalota-cprr} \fi for completeness. We now leverage Lemma~\ref{lem:PI-CPRR} to complete the proof of Proposition~\ref{prop:rev-refund-decreases-ineq}.

\fi

\ifarxiv

\begin{proof}
\ifarxiv
For the collected toll revenues, we construct a special case of the revenue refunding scheme from Lemma~\ref{lem:PI-CPRR}. In particular, consider the refunding scheme where $\alpha_g = \frac{d_g}{\sum_{g \in \G} d_g}$, which gives the refund
\ifarxiv
\begin{align*}
    r_g = \mu^g(\ttau,\0) - \mu^g(\0,\0) + \frac{1}{\sum_{g \in \G} d_g}(C_{\0} - C_{\ttau})
\end{align*}
\else 
$r_g = \mu^g(\ttau,\0) - \mu^g(\0,\0) + \frac{1}{\sum_{g \in \G} d_g}(C_{\0} - C_{\ttau})$ 
\fi
to each user in group $g$. We now show that under this revenue refunding scheme, the ex-post income distribution $\bm{\Hat{q}} = \q(\ttau, \r)$ has a lower wealth inequality measure relative to the untolled user equilibrium ex-post income distribution $\Tilde{\q} = \q(\0,\0)$. That is, we show that $W(\mathbf{\Hat{q}}) \leq W(\Tilde{\q})$.
\else
Consider the refunds $r_g = \mu^g(\ttau,\0) - \mu^g(\0,\0) + \frac{1}{\sum_{g \in \G} d_g}(C_{\0} - C_{\ttau})$ for each user in group $g$. Through an argument similar to that in~\cite[Theorem 1]{GUO2010972}, it can be shown that the corresponding CPRR scheme is user-favorable, which we present in the extended version of this paper~\cite{jalota-cprr}. We now show that under this revenue refunding scheme, the ex-post income distribution $\bm{\Hat{q}} = \q(\ttau, \r)$ has a lower wealth inequality measure relative to the untolled user equilibrium ex-post income distribution $\Tilde{\q} = \q(\0,\0)$. That is, we show that $W(\mathbf{\Hat{q}}) \leq W(\Tilde{\q})$.
\fi

To see this, we begin by considering the ex-ante income distribution $\q^0$. Under the untolled user equilibrium, users in group $g$ incur a travel cost $\mu^g(\0,\0)$, and thus the ex-post income distribution of users in group $g$ is given by $\Tilde{q}_g = q_g^0 - \beta \mu^g(\0,\0)$, where $\beta$ is the scaling factor as in Definition~\ref{def:ex-post}. On the other hand, under the CPRR scheme $(\ttau, \r)$, the ex-post income distribution of users in group $g$ is given by
\ifarxiv
\begin{align*}
    \Hat{q}_g &= q_g^0 -\beta\left( \mu^g(\ttau,\0) -r_g\right)\\ 
    &= q_g^0 - \beta \left( \mu^g(\ttau,\0) - \left[ \mu^g(\ttau,\0) - \mu^g(\0,\0) + \frac{1}{\sum_{g \in \G} d_g}(C_{\0} - C_{\ttau}) \right] \right) \\ 
    &= q_g^0 - \beta \left( \mu^g(\0,\0) - \frac{1}{\sum_{g \in \G} d_g}(C_{\0} - C_{\ttau}) \right) \\ 
    &= \Tilde{q}_g + \beta \frac{1}{\sum_{g \in \G} d_g}(C_{\0} - C_{\ttau}),
\end{align*}
\else 
\begin{align*}
    \Hat{q}_g \! &= \! q_g^0 \! - \! \beta\left( \mu^g(\ttau,\0) \! - \! r_g\right) \! = \! \Tilde{q}_g \! + \! \beta \frac{1}{\sum_{g \in \G} d_g}(\! C_{\0} \! - \! C_{\ttau} \!),
\end{align*}
\fi
where we used that $\Tilde{q}_g = q_g^0 - \beta \mu^g(\0,\0)$\ifarxiv to derive the last equality. \else. \fi Since the above relation holds for all groups $g$, \ifarxiv we observe that \fi $\mathbf{\Hat{q}} = \Tilde{\q} + \lambda \mathbf{1}$, where $\lambda = \frac{\beta}{\sum_{g \in \G} d_g}(C_{\0} - C_{\ttau}) \geq 0$. Finally, the result that $W(\mathbf{\Hat{q}}) \leq W(\Tilde{\q})$ follows by the constant income transfer property (Section~\ref{sec:model}), establishing our claim.
\end{proof}

\else 

\fi

\ifarxiv
Proposition~\ref{prop:rev-refund-decreases-ineq} establishes the existence of a user-favorable CPRR scheme that simultaneously decreases the total system cost and reduces the wealth inequality relative to that of the untolled outcome. \fi

We now present an important consequence of this result for single O-D pair travel demand when all users have values-of-time proportional to their incomes. In this setting, we show the existence of a revenue refunding scheme that decreases the wealth inequality relative to the ex-ante income distribution. We note that this is a stronger result than Proposition~\ref{prop:rev-refund-decreases-ineq} since the wealth inequality of the ex-ante income distribution is lower than that of the ex-post income distribution for the untolled case.

\begin{corollary} [CPRR Decreases Wealth Inequality for Single O-D Pair] \label{cor:rev-refund-single-od}
Consider the setting where all users travel between the same O-D pair, i.e., $w_g=w_{g'}$ for every $g,g'\in \G$, and have values-of-time proportional to their incomes, i.e., $v_g = \omega q_g^0$ for some $\omega>0$ for each group $g$. Let $\ttau$ be road tolls such that  $C_{\ttau}\leq  C_{\0}$. Then, 
there exists a revenue refunding scheme $\rr$ such that the CPRR scheme $(\ttau, \rr)$ is user-favorable and $W(\q(\ttau, \rr)) \leq W(\q^0)$, where $\q^0$ is the ex-ante income distribution.
\end{corollary}

\ifarxiv
\begin{proof}
Consider the same user-favorable CPRR scheme $(\ttau, \rr)$ as is the proof of Proposition~\ref{prop:rev-refund-decreases-ineq}. \ifarxiv We now show that the wealth inequality of the ex-post income distribution resulting from $(\ttau, \rr)$ is lower than the wealth inequality of the ex-ante income distribution, i.e., $W(\Hat{\q}) \leq W(\q)$, where $\Hat{\q} = \q(\ttau, \rr)$. \else We now show that $W(\Hat{\q}) \leq W(\q)$, where $\Hat{\q} = \q(\ttau, \rr)$. \fi \ifarxiv To see this, we first show that $W(\q(\0, \0)) = W(\q^0)$, i.e., the wealth inequality measure of the ex-ante income distribution is exactly equal to that of the untolled ex-post income distribution $\Tilde{\q} = \q(\0, \0)$. The proof of this result lies in the key observation that for any $\0$-equilibrium flow $\f^0$ all users incur the same travel time, denoted as $\gamma$, since they travel between the same O-D pair. \else To see this, we first show that $W(\q(\0, \0)) = W(\q^0)$, which follows from the observation that for any $\0$-equilibrium flow $\f^0$ all users incur the same travel time, denoted as $\gamma$, since they travel between the same O-D pair. \fi This observation leads to a travel cost of $\mu^g(\0,\0) = \omega q^0_g \gamma$ for each group $g$. Then, for the untolled setting, the ex-post income distribution of users in group $g$ is given by 
\begin{align*}
    \Tilde{q}_g = q_g^0 - \beta \mu^g(\0, \0) 
    =q_g^0  - \beta \omega q_g^0 \gamma 
    = q_g^0(1- \beta \omega \gamma).
\end{align*}
From the above, it follows that $\Tilde{\q} = \lambda_1 \q^0$ for $\lambda_1 = 1- \beta \omega \gamma$. Thus, for $\beta$ small enough it holds that $\lambda_1>0$. Under this condition, due to the scale-independence property (Section~\ref{sec:model}) of the wealth-inequality measure it follows that $W(\Tilde{\q}) = W(\q^0)$. Finally, since $W(\Hat{\q}) \leq W(\Tilde{\q})$ by the proof of Proposition~\ref{prop:rev-refund-decreases-ineq} it follows that $W(\Hat{\q}) \leq W(\Tilde{\q}) = W(\q^0)$, which proves our claim. 
\end{proof}
\fi

\ifarxiv
The above result shows that in some scenarios, 
any form of road tolls that decreases the total system cost $C_{\ttau}$ relative to $C_{\0}$, coupled with the appropriate revenue refunding policy, will not increase the level of wealth inequality in comparison to the ex-ante  distribution $\q^0$. This result \fi indicates that the appropriate refunding can reverse the 
negative consequences of tolls on wealth inequality, as was established in the ``Inequity Theorem'' in \cite{gemici_et_al:LIPIcs:2019:10270}. In particular, the ``Inequity Theorem'' asserts that for the setting considered in Corollary~\ref{cor:rev-refund-single-od}, any form of road tolls increases the level of wealth inequality compared with the ex-ante income distribution $\q^0$ in the absence of refunds.

A main ingredient in  Corollary~\ref{cor:rev-refund-single-od} is the fact that the wealth inequality measure of the ex-ante income distribution $\q^0$ is equal to that of the ex-post income distribution under the untolled user equilibrium. This result holds when users travel between the same O-D pair and have values-of-time that scale proportionally with their incomes. However, it does not hold in general for users travelling between different O-D pairs, since in such a case, users may incur different travel times at the untolled user equilibrium. \ifarxiv For the multiple O-D pair setting, we show in Proposition~\ref{prop:increase-inequality} that there are travel demand instances when no CPRR scheme can reduce income inequality relative to that of the ex-ante income distribution. \fi

\ifarxiv

\begin{proposition} [Increase in Income Inequality for Multiple O-D Pairs] \label{prop:increase-inequality}
For any user-favorable CPRR scheme $(\ttau, \rr)$, there exists a two O-D pair setting where the wealth inequality of the ex-post income distribution increases relative to the ex-ante income distribution, i.e., $W(\q(\ttau, \rr))\geq W(\q^0)$.
\end{proposition}

\begin{proof}
We show that there exists a two O-D pair setting such that for any user-favorable CPRR scheme $(\ttau,\r)$ it holds that $W(\q(\ttau,\r))\geq W(\q^0)$.

We begin by defining the instance depicted in Figure~\ref{fig:prop-2-countereg}. Consider a graph with four nodes, $v_1, v_2, v_3, v_4$ and three edges $e_1 = (v_1, v_2)$, $e_2 = (v_3, v_4)$ and $e_3 = (v_3, v_4)$, where there are two possible ways to get from $v_3$ to $v_4$. We define the travel time on edge $e_1$ as $t_1(x_1) = \frac{x_1}{2}$, that on edge $e_2$ as $t_2(x_2) = x_2$, and that on edge $e_3$ as $t_3(x_3) = 1$. Further consider two user types, one with a high income $q_H$ and value-of-time $\omega q_H$ that makes trips between O-D pair $w_H = (v_1, v_2)$, and the other with a low income $q_L$ and value-of-time $\omega q_L$ that makes trips between O-D pair $w_L = (v_3, v_4)$. Let the demand of the high income users be $d_H = 1$ and that of the low income users be $d_L = 1$. Then at the untolled user equilibrium outcome, all high income users traverse their only edge $e_1$, while all the low income users traverse the edge $e_2$. At this equilibrium flow, the cost for the high income users is $\omega q_H \frac{1}{2}$, since the travel time on edge $e_1$ is $\frac{1}{2}$, and that for the low income users is $\omega q_L$, since the travel time on edge $e_2$ is one.

\begin{figure}[!ht]
      \centering
      \includegraphics[width=0.6\linewidth]{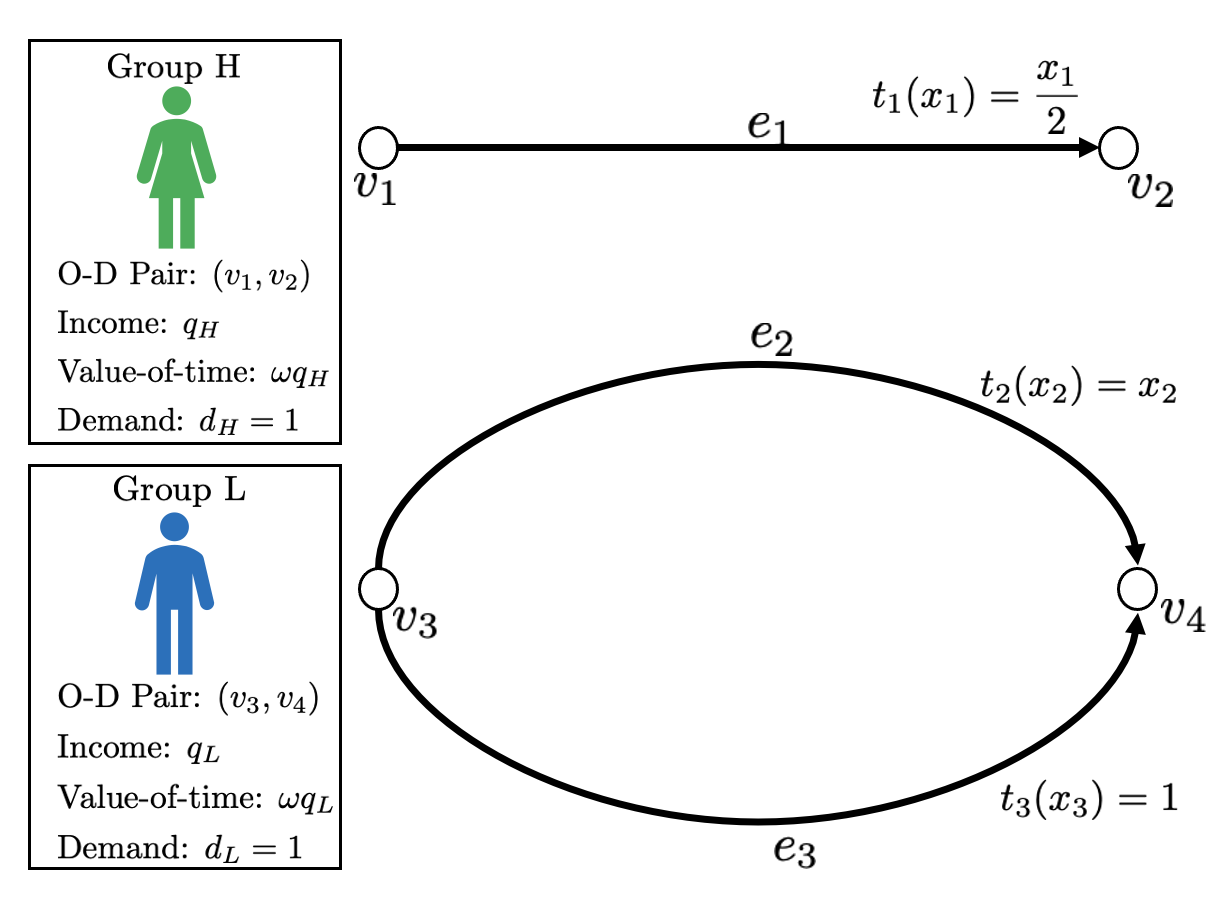}
      \caption{{\small \sf A two O-D pair and two user group instance for which the wealth inequality of the ex-post income distribution under any CPRR scheme is at least the wealth inequality of the ex-ante income distribution. Under the above defined attributes for the different user groups, the user group $H$ with a higher income and value-of-time incurs a strictly lower cost as a proportion of their income as compared to user group $L$, indicating the regressive nature of any valid CPRR scheme.}}
      \label{fig:prop-2-countereg}
   \end{figure}

Next, we note that under any CPRR scheme $(\ttau, \rr)$ users in the high income group will continue to use edge $e_1$ since this is the only available edge on which they can travel. Thus, for this scheme to be user-favorable it must be that any tolls collected from the high income users is directly refunded back within the groups. To see this, if there were tolls collected from high-income users that were given to low income users, then some high income users would incur strictly higher costs than at the untolled $\0$-equilibrium outcome. We similarly observe that all collected refunds from the low income groups must be completely refunded to users within the low income group to ensure that the CPRR scheme is user-favorable. Note that the above argument stems from the fact that the travel paths of the two user groups are completely disjoint.

Thus, we have for any user-favorable CPRR scheme $(\ttau, \rr)$ that all the users incur the same costs as that under the $\0$-equilibrium outcome. Now, under the untolled user equilibrium, we observe that the ex-post income of the high income group is $q_H = q_H - \beta \frac{\omega q_H}{2} = q_H(1-\beta \frac{\omega}{2})$ and the ex-post income of the low income group is $q_L = q_L - \beta \omega q_L = q_L(1-\beta \omega)$. The above analysis implies that the untolled user equilibrium outcome results in a regressive tax, i.e., lower income users are charged a greater fraction of their incomes than higher income users. Since the function $W$ satisfies the property that regressive taxes increase inequality, we have that the wealth inequality of the ex-post income distribution is greater than that of the ex-ante income distribution.
\end{proof}


Proposition~\ref{prop:increase-inequality} shows that there may be multiple O-D pair instances when it may not be possible to achieve a lower wealth inequality measure relative to the ex-ante income distribution. As a result, for the rest of this paper we devise CPRR schemes that reduce the wealth inequality measure relative to the ex-post income distribution under the untolled user equilibrium outcome rather than relative to the ex-ante income distribution. Note that doing so is reasonable, since we look to design CPRR schemes that improve on the status quo traffic pattern, which is typically described by the untolled user equilibrium setting. Finally, we note that it may also be interesting to develop worst-case bounds on the extent to which income inequality can be worsened through CPRR schemes in a multiple O-D pair setting relative to that of the ex-ante income distribution, and we defer this exploration to future research.

\fi

\section{Optimal CPRR Schemes} \label{sec:optimality}
In the previous section, we established the existence of a user-favorable CPRR scheme that simultaneously reduces total system cost without increasing wealth inequality relative to an  untolled outcome. In this section, we prove the existence of optimal CPRR schemes that achieve total system cost and wealth inequality that cannot be improved by any other user-favorable CPRR scheme. 
In particular, we establish that the optimal CPRR schemes are those that induce exogenous equilibrium flows with the minimum total system cost while also resulting in ex-post income distributions with the lowest level of wealth inequality among the class of all user-favorable CPRR schemes (see \ifarxiv Figure~\ref{fig:combined}\fi). We further show in Section~\ref{sec:opt-CPRR} that these optimal CPRR schemes induce equilibria even when coalitions of users endogenize the effect of refunds on their travel decisions.

We first present the main result of this section, which characterizes the set of optimal CPRR schemes.

\begin{theorem} [Optimal CPRR Scheme] \label{thm:opt}
There exists a user-favorable CPRR scheme $(\ttau^*, \rr^*)$ such that for any user-favorable CPRR scheme $(\ttau, \rr)$ it holds that $C_{\ttau^*}\leq C_{\ttau}$ and $W(\q(\ttau^*, \rr^*)) \leq W(\q( \ttau, \rr))$. 
\end{theorem}

The proof of this theorem is constructive, in the sense that it provides a recipe for computing the optimal CPRR scheme $(\ttau^*, \rr^*)$. The proof relies on two intermediate results that are of independent interest. \ifarxiv First, the ex-post income distribution of any user-favorable CPRR scheme is the same as the ex-post income distribution under the untolled user equilibrium outcome plus some non-negative transfer, which may vary across groups. We summarize this observation in the following lemma. \fi 

\begin{lemma} [Ex-post Income Distribution] \label{obs:ex-post-income}
Let $\ttau$ be road tolls such that $C_{\ttau}\leq C_{\0}$. Then, under any refunds $\rr$ such that the CPRR scheme $(\ttau, \rr)$ is user-favorable, the ex-post income of any user in group $g$ is $q_g(\ttau, \rr) = q_g(\0, \0) + \beta c_g$, where the transfer value $c_{g} \geq 0$ and satisfies the relation $\sum_{g \in \G} c_{g} d_{g} = C_{\0} - C_{\ttau}$.
\end{lemma}

\begin{proof}
Denote the ex-post income of group $g$ as $\Hat{q}_g = q_g(\ttau, \r)$. We now prove the ex-post income relation using the definition of a user-favorable CPRR scheme. In particular, for any user-favorable CPRR scheme $(\ttau, \r)$ the user travel cost does not increase from the untolled case, i.e., $\mu^g(\ttau, \r) \leq \mu^g(\0,\0)$. As it holds that $\mu^g(\ttau, \r) = \mu^g(\ttau, \0) - r_g$, we observe that for some $c_g \geq 0$ the following relation must hold for each user in group $g$: $\mu^g(\ttau, \0) - r_g + c_g = \mu^g(\0,\0)$. Then, for an ex-ante income distribution $\q^0$, the ex-post income of each user belonging to group $g$ is given by
\ifarxiv
\begin{align*}
    \Hat{q}_{g} &= q_g^0 - \beta \left( \mu^g(\ttau, \0) - r_g \right) = q_g^0 - \beta \mu^g(\0,\0) + \beta c_{g}  = q_{g}(\0, \0) + \beta c_{g},
\end{align*}
where the second equality follows since $\mu^g(\ttau, \0) - r_g = \mu^g(\0,\0) - c_g$ and the last equality follows from the observation that the ex-post income of users in group $g$ for the untolled setting is given by $q_{g}( \0, \0) = q_g^0 - \beta \mu^g(\0,\0)$.
\fi

Next, to show that $\sum_{g \in \G} c_{g} d_{g} = C_{\0} - C_{\ttau}$ we characterize the quantities $C_{\0}$ and $C_{\ttau}$. In particular, observe that by definition $C_{\0}=C(\f^0)$ and $C_{\ttau}=C(\f)$, where $\f^0$ is the untolled $\0$-equilibrium and $\f$ is an exogenous $\ttau$-equilibrium.  Now, note that both flows $\f^0$ and $\f$ can be expressed in closed form. In particular, for a given congestion-pricing scheme $\ttau'$ the exogenous $\ttau'$-equilibrium $\h(\ttau')$ can be written as 

\small
\begin{equation} \label{eq:convex-prog}
    \h(\ttau') \! = \! \argmin_{\h' \in \Omega} \! \sum_{e \in E} \! \int_{0}^{x(\h')_{e}} \! t_{e}(\omega) d \omega \! + \! \sum_{e \in E} \! \sum_{g \in \G} \! \frac{1}{v_{g}} x(\h')_{e}^{g} \tau_{e},
\end{equation}
\normalsize
where $\x(\f')$ denotes the edge representation of a path flow $\f'$. We note that this program corresponds to the \emph{multi-class user-equilibrium optimization problem}~\cite{YANG20041}.

Given this representation of the flow $\h(\ttau')$, we derive the following relation that relates the total system cost $C_{\ttau'}$ to the amount of collected revenues, by analyzing the KKT conditions of this minimization problem. In particular, it holds that
\begin{align} \label{eq:total-cost-relation}
    C_{\ttau'} = \sum_{g \in \G} \mu^{g}(\ttau', \0) d_g - \sum_{e \in E} \tau'_e x(\h(\ttau'))_e.
\end{align} 
Note that the edge flow  $\x(\h(\ttau'))$ is unique by the strict convexity of the travel-time function. We defer the proof of \ifarxiv Equation~\eqref{eq:total-cost-relation} to Appendix~\ref{apdx:kkt-multiclass-ue}\fi.

We now leverage \ifarxiv Equation~\eqref{eq:total-cost-relation} \fi to obtain that $C_{\ttau} = \sum_{g \in \G} \mu^g(\ttau,\0) d_g - \sum_{e \in E} \tau_e x(\f)_e$, where $\x(\f)=\x(\h(\ttau))$. Furthermore, from \ifarxiv Equation~\eqref{eq:total-cost-relation} \fi for the untolled setting, we obtain that $C_{\0} = \sum_{g \in \G} \mu^g(\0,\0) d_g$. Finally, using these two relations and leveraging the fact that $c_g = \mu^g(\0,\0) - \mu^g(\ttau, \0) + r_g$ we get
\ifarxiv
\begin{align*}
    \sum_{g \in \G} c_{g} d_{g} &= \sum_{g \in \G} (\mu^g(\0,\0) - \mu^g(\ttau, \0) + r_g) d_{g}, \\
    &= \sum_{g \in \G} \mu^g(\0,\0) d_g - \sum_{g \in \G} \mu^g(\ttau, \0) d_g + \sum_{g \in \G} r_{g} d_g, \\
    &\stackrel{(a)}{=} C_{\0} - \sum_{g \in \G} \mu^g(\ttau, \0) d_g + \sum_{e \in E} \tau_e x(\f)_e, \\
    &\stackrel{(b)}{=} C_{\0} - C_{\ttau},
\end{align*}
\fi
where (a) follows as $\sum_{g \in \G} r_{g} d_g = \sum_{e \in E} \tau_e x(\f)_e$ and $C_{\0} = \sum_{g \in \G} \mu^{g}(\0, \0)$, and (b) follows as $C_{\ttau} = \sum_{g \in \G} \mu^g(\ttau, \0) d_g - \sum_{e \in E} \tau_e x(\f)_e$. This proves our claim.
\end{proof}
\ifarxiv
The above lemma highlights that under any user-favorable CPRR scheme $(\ttau, \rr)$, each user's income is at least their ex-post income under the untolled case.
\fi

The second result required to prove Theorem~\ref{thm:opt} relies on the observation that there is a monotonic relationship between the minimum achievable wealth-inequality measure and the total system cost.

\begin{lemma} [Monotonicity of Refunds] \label{obs:monotonicity-refunds}
Suppose that there are two congestion-pricing schemes $\ttau_A$ and $\ttau_B$ with total system costs satisfying $C_{\ttau_A}\leq C_{\ttau_B} \leq C_{\0}$. Then there exists a revenue refunding scheme $\rr_A$ such that $(\ttau_A, \rr_A)$ is user-favorable and achieves a lower wealth inequality measure than any user-favorable CPRR scheme $(\ttau_B, \rr_B)$ for any revenue refunds $\rr_B$, i.e., $W(\q(\ttau_A,\rr_A))\leq W(\q(\ttau_B,\rr_B))$.
\end{lemma}

\begin{proof}
We prove this claim by constructing for each revenue refunding scheme $\r_B$ under the tolls $\ttau_B$, a revenue refunding scheme $\r_A$ under the tolls $\ttau_A$ that achieves a lower wealth inequality. To this end, we first introduce some notation. Let $c_g^A$ and $c_g^B$ be non-negative transfers for each group $g$ as in Lemma~\ref{obs:ex-post-income}, where $\sum_{g \in \G} c_g^A d_g = C_{\0} - C_{\ttau_A}$ and $\sum_{g \in \G} c_g^B d_g = C_{\0} - C_{\ttau_B}$ must hold for the feasibility of the scheme.

Then, by Lemma~\ref{obs:ex-post-income} we have that the ex-post income of users in group $g$ can be expressed as: $q_{g}(\ttau_A, \r_A) = q_g(\0, \0) + \beta c_{g}^A$ and $q_{g}(\ttau_B, \r_B) = q_g(\0, \0) + \beta c_{g}^B$. Let $c_{g}^A = c_{g}^B + \frac{1}{\sum_{g \in \G} d_{g}}(C_{\ttau_B} - C_{\ttau_A})$. We now show that the refunding $\r_A$ is feasible by observing that 
\ifarxiv
\begin{align*}
    \sum_{g \in \G} c_{g}^A d_{g} &= \sum_{g \in \G} \left(c_{g}^B d_{g} + \frac{d_{g}}{\sum_{g \in \G} d_{g}}\left(C_{\ttau_B} - C_{\ttau_A}\right) \right), \\ &= \sum_{g \in \G} c_{g}^B d_{g} + C_{\ttau_B} - C_{\ttau_A}, \\
    &= C_{\0} - C_{\ttau_B} + C_{\ttau_B} - C_{\ttau_A} = C_{\0} - C_{\ttau_A},
\end{align*}
\fi
Here we leveraged the fact that $\sum_{g \in \G} c_{g}^B d_{g} = C_{\0} - C_{\ttau_B}$. 

Under the above defined non-negative transfer $c_g^A$, we observe that the ex-post income distribution under the CPRR scheme $(\ttau_A, \r_A)$ is the same as the ex-post income distribution under the CPRR scheme $(\ttau_B, \r_B)$ plus a constant positive transfer, which is equal for all users. That is, we have $\q(\ttau_A, \r_A) = \q(\ttau_B, \r_B) + \lambda \mathbf{1}$ for $\lambda = \frac{\beta}{\sum_{g \in \G} d_{g}}(C_{\ttau_B} - C_{\ttau_A}) \geq 0$. Finally, by the constant income transfer property (Section~\ref{sec:model}) it follows that $W(\q(\ttau_A, \r_A)) \leq W(\q(\ttau_B, \r_B))$.
\end{proof}

The above result establishes a very natural property of any user-favorable revenue-refunding policy for which the total refund remaining after satisfying the user-favorable condition is $C_{\0} - C_{\ttau}$. In particular, a smaller total system cost yields a larger amount of remaining refund $C_{\0} - C_{\ttau}$, which, in turn, results in a greater degree of freedom in distributing these refunds to achieve an overall lower level of wealth inequality.

Finally, Theorem~\ref{thm:opt} follows directly by the monotonicity relation established in Lemma~\ref{obs:monotonicity-refunds}, and prescribes a two-step procedure to find a optimal CPRR scheme that is also user-favorable. In particular, choose a congestion pricing scheme $\ttau^*$ such that the total travel cost is minimized, i.e., $C_{\ttau^*}=C^*$. Next, select the revenue refunding scheme $\rr^*$ to be such that the expression $W(\q(\ttau^*,\rr^*))$ is minimized and $(\ttau^*,\rr^*)$ is user-favorable through an appropriate selection of transfers $c_g$. For more details on computing an optimal CPRR scheme, we refer to Sections~\ref{sec:optTolls} and~\ref{sec:opt-scheme-kkt}. \ifarxiv Now, let $(\ttau,\rr)$ be some user-favorable CPRR scheme. By definition of $\ttau^*$, it holds that $C_{\ttau^*}\leq C_{\ttau}$. Moreover, Lemma~\ref{obs:monotonicity-refunds} ensures that $W(\q(\ttau^*,\rr^*))\leq W(\q(\ttau,\rr))$ is satisfied. \fi

\ifarxiv
\paragraph{Significance of Theorem~\ref{thm:opt}.}
\fi
Theorem~\ref{thm:opt} establishes that the optimal CPRR scheme is one that simultaneously achieves the highest efficiency whilst also reducing wealth inequality to the maximum degree possible among \ifarxiv the class of \fi all user-favourable CPRR schemes. This finding is counter-intuitive since equity and efficiency are typically at odds but Theorem~\ref{thm:opt} establishes that no such tradeoff between system efficiency and wealth inequality exists. The reason for this is that the remaining refund after satisfying the user-favourable condition increases as the total system cost decreases (Lemma~\ref{obs:monotonicity-refunds}), thereby giving greater leverage in the design of the refunding mechanism to achieve a lower wealth inequality. We further present numerical experiments in Appendix~\ref{apdx:numerical_main} to demonstrate the efficacy of optimal CPRR schemes and show that the benefits of CPRR can even be realized when users’ values of time are not exactly known to the central planner.

\section{Computational and Equilibrium Properties of Optimal CPRR Schemes} \label{sec:opt-CPRR}  

Having established the existence of optimal CPRR schemes, we now show how such schemes can be computed and highlight additional equilibrium properties of these schemes. To this end, in Sections~\ref{sec:optTolls} and~\ref{sec:opt-scheme-kkt}, we provide a concrete recipe for computing the optimal CPRR scheme $(\ttau^*,\rr^*)$ for a commonly used wealth inequality measure, the \emph{discrete Gini coefficient}\ifarxiv, when users have finitely many incomes levels~\cite{WU20121273}. \fi Then, in Section~\ref{sec:endogenous}, we consider the endogenous equilibrium setting, wherein users minimize a linear function of not only their travel times and tolls but also \ifarxiv their \fi refunds. In this setting, we show that the optimal CPRR scheme is robust to user coalitions, i.e., optimal CPRR schemes induce equilibria even when \emph{coalitions} of users endogenize the effect of refunds on their travel decisions.

\subsection{Computing Optimal Tolls} \label{sec:optTolls}

The problem of computing optimal tolls $\ttau^*$ has been widely studied in the literature\ifarxiv~\cite{YANG20041,METRP,hearn1998solving}\fi. In particular,~\cite{YANG20041} showed by analysing the first order KKT conditions of the minimum total system cost problem, presented in Section~\ref{sec:metrics}, that the optimal toll on each edge is given by 
\ifarxiv
\begin{align} \label{eq:optTolls}
    \tau_e^* = \left( \sum_{g \in \G} \frac{x_{e}^g}{x_e} v_g \right) x_e t_e'(x_e), \quad \forall e \in E,
\end{align}
where edge flows $\x$ and the group specific edge flows $\x^g$ correspond to the edge decomposition of the optimal path flow $\f$ of the minimum total system cost problem: $\f = \argmin_{\f\in \Omega}C(\f)$. Observe that the optimal tolls in \ifarxiv Equation~\eqref{eq:optTolls} \else \eqref{eq:optTolls} \fi to minimize the total system cost is akin to marginal cost prices, given by $x_e t_e'(x_e)$ for each edge $e$, when all users are homogeneous with the same values of time. In particular, note that the toll on each edge $e$ in \ifarxiv Equation~\eqref{eq:optTolls} \else \eqref{eq:optTolls} \fi is given by the the travel time externality, i.e., the marginal cost prices of users multiplied by the average value of time of users on that edge.
\fi


\subsection{Computing Optimal Revenue Refunds \ifarxiv for the Discrete Gini Coefficient \fi} \label{sec:opt-scheme-kkt}

Given the method to compute optimal tolls $\ttau^*$, as elucidated in the previous section, we now focus our attention on deriving the optimal revenue refunding policy $\rr^*$ for a commonly used wealth inequality measure, the discrete Gini coefficient. In particular, we show in this section that the optimal revenue refunding scheme for the discrete Gini coefficient measure corresponds to a natural max-min refunding scheme wherein the refunds are given to users belonging to the lowest income groups.

We first present the discrete Gini coefficient measure and discuss some of its properties.

\begin{definition} [Discrete Gini Coefficient~\cite{WU20121273}]
Let the mean income corresponding to the income distribution $\q$ with a demand vector $\ddb = \{d_g: g \in G \}$ be $\Delta(\q) = \frac{\sum_{g \in \G} q_g d_g}{\sum_{g \in \G} d_g}$. Then, the discrete Gini coefficient $W$ is given by $W(\q) = \frac{1}{2 \left(\sum_{g \in \G} d_g \right)^2 \Delta(\q)} \sum_{g_1, g_2 \in \G} d_{g_1} d_{g_2} |q_{g_1} - q_{g_2}|$.
\end{definition}

\ifarxiv
A few comments about the discrete Gini coefficient as a wealth inequality measure are in order. First, the discrete Gini coefficient is zero if all users have the same income, i.e., there is perfect equality in society. Next, due to the absolute value of the difference between user incomes in the numerator, the discrete Gini coefficient is larger if the dispersion of incomes between different user groups is greater. Finally, note here that we do not write the discrete Gini coefficient measure as a function of the vector of demands $\ddb = \{d_g: g \in G \}$ since we assume that user demands are fixed in this work.

\begin{remark}
The discrete Gini coefficient measure satisfies the scale independence, regressive and progressive tax properties required for it to be a valid wealth inequality measure $W$. To see this, note that scaling all incomes by a positive constant $\lambda>0$ scales both the numerator and denominator by a factor of $\lambda$, which results in scale independence. The proof of the discrete Gini coefficient satisfying the regressive tax property is presented in \ifarxiv Appendix~\ref{apdx:gini-regressive}\else the extended version of our work~\cite{jalota-cprr}\fi. The discrete Gini coefficient satisfies the progressive tax property following a similar argument. Thus, it is a valid wealth-inequality measure.
\end{remark}
\fi

For the discrete Gini coefficient, we now present a mathematical program for computing the revenue refunding policy $\rr^*$. To this end, we first observe by Lemma~\ref{obs:ex-post-income} that for any user-favorable CPRR scheme $(\ttau^*, \rr^*)$ each user's ex-post income is given by $q_g(\ttau^*, \rr^*) = q_g(\0, \0) + c_g$ (where, for ease of exposition, we let $\beta = 1$) for some $c_g \geq 0$, where $\sum_{g \in \G} c_g d_g = C_{\0} - C_{\ttau^*}$. Thus, the choice of the optimal revenue refunds $\rr^*$ can be reduced to computing the optimal transfers $c_g$. In particular, we formulate the computation of the optimal transfers $c_g$ to minimize the discrete Gini coefficient through the following optimization problem:
\ifarxiv
\begin{mini!}|s|[2]                   
    {c_{g}}                              
    {W(\q(\0, \0) + \cc) 
    , \label{eq:Gini-obj}}   
    {\label{eq:Eg001}}             
    {}                                
    \addConstraint{\sum_{g \in \G} c_{g} d_{g}}{= C_{\0} - C_{\ttau^*}, \label{eq:Gini-con1}}    
    \addConstraint{c_{g}}{ \geq 0, \text{ for all } g \in \G, \label{eq:pareto-imp-refund}}
\end{mini!}
\fi
where \ifarxiv we denote $\Tilde{q}_g = q_{g}(\0, \0)$ for conciseness, $\cc = \{c_g: g \in \G \}$, and $\q(\0, \0) + \mathbf{c}$ represents the income distribution of users after receiving the revenue refunds. \fi Furthermore, noting that $\Delta(\q(\0, \0) + \mathbf{c}) = \frac{C_{\0} - C_{\ttau^*} + \sum_{g \in \G} q_g(\0, \0) d_g}{\sum_{g \in \G} d_g}$ is a fixed quantity, the above problem can be solved via a linear program (see Chapter 6 in ~\cite{bisschop2016aimms}). \ifarxiv The optimal revenue refunding policy $\rr^*$ corresponding to the above optimization problem results in a natural max-min outcome. In particular, Algorithm~\ref{alg:greedy} describes the revenue refunding process, wherein users in the lowest income groups, denoted by $\Gmin$, are provided refunds until their incomes equal that of the second lowest income groups. This process is repeated until all the collected revenue is completely exhausted and is depicted in Figure~\ref{fig:Optimal-Refunds}. \fi

\ifarxiv

\begin{algorithm} 
\SetAlgoLined
\SetSideCommentRight
\SetNoFillComment 
\SetKwInOut{Input}{Input}\SetKwInOut{Output}{Output}
\Input{Untolled exogenous equilibrium ex-post income distribution: $\Tilde{\q}:=\q(\0, \0)$, Total system cost $C_{\0}$ for the untolled setting, Total system cost $C_{\ttau}$ under the tolls $\tau$.}
\Output{Nonnegative transfers $\c = \{c_g: g \in \G \}$}
$Z \leftarrow C_{\0} - C_{\ttau}$ \quad  \tcc{Total refund assigned}  
$\c \leftarrow \0$  \quad\tcc{initialize transfers}
\While{$Z>0$}{
  $\qmin \leftarrow \min_{g \in \G}  \Tilde{q}_g$ \quad \tcc{Lowest income}
  $\Gmin \leftarrow \left\{g\in \G \middle| \tilde{q}_g=q_{\text{min}}\right\}$ \quad \tcc{Lowest income groups}
  $\qnext \leftarrow \min_{g \in \G \setminus \Gmin} \Tilde{q}_g$ \quad \tcc{Second-lowest income}
  $Y \leftarrow \min \left\{ \qnext - \qmin, \frac{Z}{\sum_{g \in \Gmin} d_{g}} \right\}$ \quad \tcc{Additional income per person}
  $\tilde{q}_g \leftarrow \tilde{q}_g + Y, \quad \forall g \in \Gmin$ \quad \tcc{Update income}
  $c_g \leftarrow c_g + Y, \quad \forall g \in \Gmin$ \quad \tcc{Update transfer}
  $Z \leftarrow Z - Y \sum_{g \in \Gmin} d_g$ \quad \tcc{Update remaining refund}
  }
\caption{Max-Min Revenue Refunding}
\label{alg:greedy}
\end{algorithm}
The following proposition establishes that the optimal vector of non-negative transfers corresponding to the solution to the optimization problem given by~\eqref{eq:Gini-obj}-\eqref{eq:pareto-imp-refund} is equal to the vector of transfers $\cc$ computed through Algorithm~\ref{alg:greedy}.

\begin{proposition} [Optimal Revenue Refunding Scheme] \label{prop:opt-rev-refunds-discrete-gini}
The vector of transfers $\cc = \{c_g: g \in \G \}$ output by Algorithm~\ref{alg:greedy} is equal to the optimal vector of transfers computed through the solution to the optimization problem given by~\eqref{eq:Gini-obj}-\eqref{eq:pareto-imp-refund}.
\end{proposition}

For a proof of Proposition~\ref{prop:opt-rev-refunds-discrete-gini}, see \ifarxiv Appendix~\ref{apdx:opt-rev-discrete-gini}\else the extended version of our paper~\cite{jalota-cprr}\fi. We note here that the greedy process of revenue refunding elucidated in Proposition~\ref{prop:opt-rev-refunds-discrete-gini} is reminiscent of Rawl's difference principle of giving the greatest benefit to the most disadvantaged groups of society~\cite{rawls2020theory}.

\begin{figure}[!ht]
      \centering
      \includegraphics[width=0.5\linewidth]{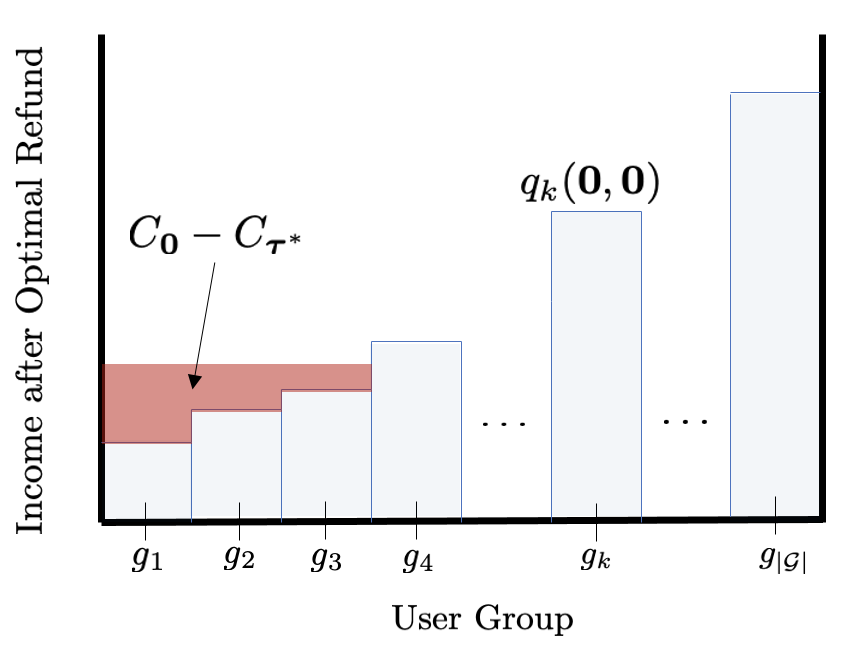}
      \caption{{\small \sf The optimal solution of Algorithm~\ref{alg:greedy} is analogous to a max-min allocation. The height of the grey bars represents the ex-post income of users under the untolled outcome, while the red region denotes the amount of transfer $c_g$ that is given to the different user groups. The total sum of the transfers given to users adds up exactly to $C_{\0} - C_{\ttau^*}$ under the set of tolls $\ttau^*$.}}
      \label{fig:Optimal-Refunds}
   \end{figure}

\fi

\subsection{\ifarxiv Optimal \fi CPRR Schemes and Endogenous Equilibria} \label{sec:endogenous}

Thus far, we have considered the setting when all users minimize a linear function of their travel times and tolls without considering refunds. In this section we consider the setting of the \emph{endogenous} equilibrium, wherein users minimize a linear function of not only their travel times and tolls but also refunds. In particular, we consider two equilibrium notions (without and with user coalitions) in this endogenous setting and show that the optimal CPRR scheme induces equilibria in both these settings. To this end, we first consider endogenous equilibria without user coalitions and show that any endogenous equilibrium is an exogenous equilibrium. Next, in the setting of endogenous equilibria \emph{with} user coalitions, we show that while, in general, endogenous equilibria do not coincide with exogenous equilibria, the optimal CPRR scheme is robust to user coalitions, i.e., any exogenous equilibrium induced by an optimal CPRR scheme is also an endogenous equilibrium.

\subsubsection{Endogenous Equilibria without User Coalitions}

\ifarxiv

We begin by considering the setting of an endogenous equilibrium without user coalitions and show that endogenous and exogenous equilibria are \ifarxiv equivalent in the absence of user coalitions. \else equivalent. \fi In this setting without user coalitions, the definition of an exogenous equilibrium (Definition~\ref{def:exo-eq}) can be readily extended to the setting when users additionally account for refunds in their travel cost minimization, as is elucidated by the following definition.

\begin{definition} [Endogenous Equilibrium without Coalitions] \label{def:endo-eq-wo-coalitions}
For a given CPRR scheme  $(\ttau, \r)$, a path flow pattern $\f$ is an endogenous equilibrium without coalitions if for each group $g\in \G$ it holds that $f_{P,g}>0$ for some path $P \in \P_{g}$ if and only if
\ifarxiv
\begin{align*}
    \mu_{P}^g(\f, \ttau,\r) \leq \mu_{Q}^g(\f, \ttau, \r), \text{ for all } Q \in \P_{g}.
\end{align*}
\else
$\mu_{P}^g(\f, \ttau,\r) \leq \mu_{Q}^g(\f, \ttau, \r), \text{ for all } Q \in \P_{g}.$ 
\fi
We say that such an $\f$ is an endogenous $(\ttau, \r)$-equilibrium without coalitions.
\end{definition}
We now show that any exogenous equilibrium is also an endogenous equilibrium without coalitions and vice versa, i.e., the two equilibrium concepts are equivalent. This result follows since we are in the setting of a non-atomic congestion game, wherein users are infinitesimal, and thus a unilateral deviation by any user will not influence their overall refunds since the flow of users in the network remains unchanged and the tolls are fixed.

\begin{observation} [Equivalence between Exogenous Equilibria and Endogenous Equilibria without Coalitions] \label{rmk:exo-equals-endo}
Any exogenous $\ttau$-equilibrium flow $\f$ corresponds to an endogenous $(\ttau, \r)$-equilibrium without coalitions for any refunding scheme $\r$. 
\end{observation}

\ifarxiv

\begin{proof}
To see this, observe by the definition of an exogenous $\ttau$-equilibrium flow $\f$ that $f_{P,g}>0$ for some path $P \in \P_{g}$ if and only if
\begin{align*}
    \mu_{P}^g(\f, \ttau,\0) \leq \mu_{Q}^g(\f, \ttau, \0), \text{ for all } Q \in \P_{g}.
\end{align*}
Subtracting both sides of the equation by $r_g$ (which cannot be influenced by users through a unilateral deviation in a non-atomic setting), we obtain that $f_{P,g}>0$ for some path $P \in \P_{g}$ if and only if
\begin{align*}
    \mu_{P}^g(\f, \ttau,\0) - r_g \leq \mu_{Q}^g(\f, \ttau, \0) - r_g, \text{ for all } Q \in \P_{g}, \\
    \implies \mu_{P}^g(\f, \ttau,\r) \leq \mu_{Q}^g(\f, \ttau, \r), \text{ for all } Q \in \P_{g},
\end{align*}
i.e., $f_{P,g}>0$ for some path $P \in \P_{g}$ if and only if $\mu_{P}^g(\f, \ttau,\r) \leq \mu_{Q}^g(\f, \ttau, \r), \text{ for all } Q \in \P_{g}$. In other words, the flow $\f$ corresponds to an endogenous $(\ttau, \r)$-equilibrium without coalitions. Observe that to establish that any endogenous $(\ttau, \r)$-equilibrium $\f$ without coalitions is an exogenous $\ttau$-equilibrium, we can simply reverse the above arguments, which establishes our claim.
\end{proof}
\else
For a proof of Observation~\ref{rmk:exo-equals-endo}, we refer to the extended version of our paper~\cite{jalota-cprr}. \fi \ifarxiv Observation~\ref{rmk:exo-equals-endo} establishes an equivalence between exogenous equilibria and endogenous equilibria without coalitions, and, in particular, shows that any exogenous equilibrium flow $\f$ induced by an optimal CPRR scheme is also an endogenous equilibrium flow without coalitions. \fi

\fi

\subsubsection{Endogenous Equilibria with User Coalitions}

In this section, we consider the stronger endogenous equilibrium notion wherein \emph{coalitions} of users minimize a linear function of not only their travel times and tolls but also refunds. In particular, we consider the setting wherein each user group is treated as a coalition. Note that unlike the setting without coalitions, in this setting, a change in the strategy of the entire group, i.e., the flow sent on each feasible path, will likely result in a change in the overall network flow and correspondingly the revenues obtained by each user in the group. In the presence of user coalitions, we show that while the notions of exogenous equilibrium and endogenous equilibrium with coalitions do not agree in general, any exogenous equilibrium induced by an optimal CPRR scheme is also an endogenous equilibrium with coalitions.

To this end, we begin by introducing the notion of an endogenous equilibrium with coalitions.

\begin{definition} [Endogenous Equilibrium with Coalitions] \label{def:endo-eq2}
Let $(\ttau,\rr)$ be a CPRR scheme, and $\f$ be a flow pattern. Then $\f$ is an endogenous $(\ttau,\rr)$-equilibrium with coalitions if for each group $g\in \G$, every path $P\in \P_{g}$ such that $f_{P,g}>0$, and any flow pattern $\f'$ such that 
\[\f'_{P',g'}=\f_{P',g'}, \text{ for all } g'\in \G\setminus \{g\}, P'\in \P_{g'},\]
it holds that 
\begin{align*}
    \mu_{P}^g(\f, \ttau,\rr(\f,\ttau)) \leq \mu_{Q}^g(\f', \ttau, \rr(\f',\ttau)), \text{ for all } Q \in \P_{g}.
\end{align*}
Here $\f'$ denotes a flow that results from $\f$ where exactly one group changes its path assignment. 
\end{definition}

A few comments about the above definition are in order. First, it is clear that the above definition of endogenous equilibrium is a stronger notion than the standard Nash equilibrium considered in non-atomic congestion games. This is because every endogenous equilibrium is a Nash equilibrium when users minimize their travel costs including refunds but not every Nash equilibrium is necessarily an endogenous equilibrium. 

Next, we restrict the set of possible coalitions to those corresponding to strategies for a given user group. This is often reasonable, since users belonging to similar income levels that make similar trips, i.e., travel between the same O-D pair, are more likely to be socially connected with each other and share travel information as compared to users across groups. As a result, we do not consider the setting of equilibrium formation that is robust to any arbitrary set of coalitions~\cite{ANDELMAN2009289}, and defer this as an interesting direction for future research. 

Furthermore, we can view the endogenous equilibrium as a non-atomic analogue of the atomic equilibrium setting, wherein each group $g$ controls a flow of $d_g$. In atomic settings, each group only sends its flow on one path, whereas in the non-atomic setting, the flows can be dispersed across multiple paths with equal travel costs.

\paragraph{Endogenous Equilibria with Coalitions Differ from Exogenous Equilibria} \label{sec:endo-exo-differ}
We first show that, in general, the endogenous equilibria with coalitions and exogenous equilibria are not the same. To this end, we first recall that an exogenous equilibrium only depends on the tolling scheme $\ttau$ and is completely independent of the refunds $\rr$. On the other hand, since users take into account revenue refunds in the case of the endogenous equilibrium, each user must know the refunding policy $\rr$ to reason about their strategies when making travel decisions. In particular, each user (and coalition of users within a group) must be able to reason about how a change in their strategy, i.e., the path(s) on which they travel, will change the total amount of refund they receive, and in effect their travel cost. Thus, for this section, we restrict our attention to revenue-refunding schemes resulting from \ifarxiv Algorithm~\ref{alg:greedy}.\fi That is, users are given refunds through a process analogous to a max-min allocation. We now construct a counterexample to show that an exogenous $\ttau$-equilibrium flow may no longer be an equilibrium when users take into account refunds in their travel cost minimization.

\begin{proposition} [Non-Equivalence of Equilibria] \label{prop:two-link-eq-countereg}
There exists a setting with (i) a two-edge parallel network, (ii) three income classes, and (iii) tolls $\ttau$, such that the induced exogenous $\ttau$-equilibrium is not an endogenous $(\ttau,\rr)$-equilibrium with coalitions, where $\rr$ results from the max-min revenue refunding policy \ifarxiv(Algorithm~\ref{alg:greedy}).\fi
\end{proposition}

For a proof of Proposition~\ref{prop:two-link-eq-countereg}, see \ifarxiv Appendix~\ref{eq:pfprop4}\fi. The above proposition is quite natural, since low-income users may take routes that were previously unaffordable when taking into account revenue refunds in their route selection process.

\paragraph{Endogenous Equilibria with Coalitions Coincide with Exogenous Equilibria at the Optimal Solution}

While Proposition~\ref{prop:two-link-eq-countereg} indicates that, in general, the exogenous equilibria and endogenous equilibria with coalitions do not coincide, we now establish that any exogenous equilibrium induced by an optimal user-favorable
CPRR scheme $(\ttau^*, \rr^*)$, where the refund satisfies a mild condition, is also an endogenous equilibrium. In particular, we have the following lemma:

\begin{lemma}
[Optimal CPRR Scheme under Endogenous Equilibria] \label{cor:sys-par-opt-endo}
Let $(\ttau^*, \rr^*)$ be an optimal user-favorable CPRR scheme under the exogenous equilibrium model and let $\f^*$ be its exogenous equilibrium. In addition, let $\f^0$ be a $\0$-equilibrium and $C^*$ be the minimum total system cost. Further, suppose that the refunding scheme $\rr^*$ is defined as $r^*_g:=\mu^g(\ttau^*, \0)-\mu^g(\0, \0)+c_g(C^*)$, where the non-negative transfer $c_g(C(\f))$ for each group $g$ is monotonically non-increasing in the total system cost $C(\f)$ for a given flow $\f$.
Then $\f^*$ is also an endogenous $(\ttau^*, \rr^*)$-equilibrium with coalitions.
\end{lemma}

\ifarxiv

\begin{proof}
As in the analysis of Lemma~\ref{obs:ex-post-income}, for any user-favorable CPRR scheme $(\ttau^*, \rr^*)$ it holds for some $c_g$ for each group $g$ that the travel cost to users in group $g$ under the exogenous $\ttau^*$-equilibrium $\f^*$ is given by $r^*_g=\mu^g(\ttau^*, \0)-\mu^g(\0, \0)+c_g$, where $c_g \geq 0$ and $\sum_{g \in \G} c_g d_g = C_{\0} - C^*$.

We now consider the emerging behavior of users for the endogenous setting. Since $\mu^g(\0,\0)$ is a fixed quantity representing the travel cost at the untolled $\0$-equilibrium $\f^0$, the best response of any coalition within a group $g$ under the endogenous equilibrium, when minimizing each user's individual travel cost $\mu^g(\0,\0) - c_g$ (see the analysis in Lemma~\ref{obs:ex-post-income}), is to maximize $c_g$. 

Next, since for each user group $g$, $c_g$ is monotonically non-decreasing in $C_{\0} - C(\f)$, we have that $c_g$ is maximized for each user group $g$ when $C_{\0} - C(\f)$ is maximized. Since $C_{\0}$ is fixed, $C_{\0} - C(\f)$ is maximized for any flow $\f^*$ with the minimum total system cost $C^*$. This implies that each user's non-negative transfer $c_g$ is maximized for any flow $\f^*$ with the minimum total system cost. Thus, any exogenous $\ttau^*$-equilibrium flow $\f^*$ that achieves the minimum total system cost is also an endogenous equilibrium with coalitions, since a deviation by any coalition of users in group $g$ can never result in a higher non-negative transfer $c_g$ than that at the minimum total system cost solution.
\end{proof}

\fi

\ifarxiv

Lemma~\ref{cor:sys-par-opt-endo} establishes that any exogenous equilibrium flow arising from an optimal CPRR scheme is also an endogenous equilibrium with coalitions. The condition that the non-negative transfer $c_g$ for any group $g$ is monotonically non-increasing in the total system cost is not demanding. For instance, the optimal refunding scheme, i.e., the one minimizing wealth inequality, corresponding to the discrete Gini coefficient respects this montonicity relation, \ifarxiv as observed in Proposition~\ref{prop:opt-rev-refunds-discrete-gini}. \else as described in Section~\ref{sec:opt-scheme-kkt}. \fi

\fi

\ifarxiv

\section{Discussion} \label{sec:discussion}
A core tenet of sustainable transportation entails achieving a balance between \emph{economic}, \emph{equity} and \emph{environmental} goals~\cite{hall_promoting_2006}. The results demonstrated in this paper challenge the traditional notion that these goals are in tension with each other by making progress towards achieving each of these goals simultaneously. From an economic perspective, we not only show that our CPRR schemes minimize the total system cost, but also that every user is left at least as well off under the CPRR schemes as compared to that prior to any implementation of congestion pricing or refunds. As a result, not only do our schemes achieve the desired economic outcomes at a societal scale, but also all users would \emph{favor} our proposed pricing and refunding schemes. Next, from the perspective of equity, we show that the gains achieved for each user through these schemes are distributed in such a way that the overall wealth inequality in society is reduced. Finally, from an environmental perspective, the total system cost objective, which we seek to minimize within optimal CPRR schemes (Theorem~\ref{thm:opt}), can be treated as an imperfect proxy for the total environmental pollution in the system. This is because the environmental impact of a scheme is often proportional to the total travel time of all users. Environmental goals can be more directly incorporated into the design of CPRR schemes  through further model refinements to explicitly incorporate environmental costs such as those for the air pollution impacts of CPRR schemes, beyond just the total system cost.

Our work demonstrated that if we look at congestion pricing from the lens of refunding the collected tolls, we can not only achieve system efficiency but also reduce wealth inequality. As a result, we view our work as a significant step in shifting the discussion around congestion pricing from one that has focused on the inequity impacts of road tolls to one that centers around how to best distribute the revenues collected to different sections of society. While refunding toll revenues is not novel and has been proposed as early as in~\cite{small1992using}, our work provided a thorough characterization of how such schemes can be designed to simultaneously achieve system efficiency and equity objectives. Furthermore, in doing so, we ensured that all users are at least as well off as compared to before the introduction of the CPRR scheme, thereby making it publicly acceptable to all users. 

We believe that while our results pave the way for the design of sustainable, publicly-acceptable congestion-pricing schemes, significant practical challenges remain. For instance, we assume centralized knowledge of the values-of-time of each user group. In practice these may not be known, and could confound successful implementation of an optimal CPRR scheme. As a result, it is important to leverage empirical estimation methods of users' preferences~\cite{arora2020private,Buchholzrdab050,cohen2016using} to better inform the inputs that are necessary for the design of optimal CPRR schemes. Furthermore, it may not be possible to directly refund all the collected toll revenues since it may need to be used for other purposes, e.g., subsidies to improve the transit infrastructure. It is also important to note the degree to which the CPRR scheme is successful relies on the full implementation of the tolls \emph{and} refunds. If policymakers implement the congestion pricing scheme but fail to deliver refunds, low-income users of the system will be made worse off, facing higher costs, worse travel times, or both. Underprivileged residents would have legitimate claims that the system was not working, undermining public trust in the system. Thus, the onus is on policy makers to manage the entire life cycle of the CPRR scheme and ensure its successful and sustainable implementation. The difference between an equitable, optimal congestion pricing scheme and one that disproportionately burdens the poor depends significantly on how the toll revenue is spent.

\section{Conclusion} \label{sec:conclusion}

In this paper, we studied and designed user-favorable congestion pricing and revenue refunding (CPRR) schemes that mitigate the regressive wealth inequality effects of congestion pricing. In particular, under the user behavior model when users minimize a linear function of their travel times and tolls, without considering refunds, we developed CPRR schemes that improve both system efficiency and wealth inequality, while being favorable for all users, as compared to the untolled outcome. We further characterized the set of optimal CPRR schemes and presented a method to compute the optimal scheme when the wealth inequality measure is given by the discrete Gini coefficient. Finally, we showed that even when (coalitions of) users endogenize the effect of refunds on their travel decisions, the resultant equilibrium remains the same under any optimal CPRR scheme. 

There are several interesting directions for further research. The first would be to relax some of the commonly-used assumptions in transportation research and game theory. One example is to consider nonlinear user travel cost functions. In addition, we currently assume time-invariant travel demand and traffic flows, which motivates the possible generalization to dynamic settings, e.g., through the incorporation of the cell transmission model~\cite{daganzo1994cell}. We have also assumed that the only decisions made by users are route choices, whereas in reality there are other options, such as changing departure time or travel mode. A possible way to overcome this limitation is by incorporating elastic-demand models into our traffic-assignment formulations\ifarxiv~\cite{Patriksson15,gartner1980optimal}\else~\cite{Patriksson15}\fi. Finally, we only consider direct refunds to the road users. We can extend our framework to analyze system designs with cross subsidies across multiple forms of transport, e.g., subsidies to improve the transit infrastructure.

It would also be interesting to consider a more general class of refunding mechanisms wherein some portion of the collected revenues is used to cover operational costs or improve transportation infrastructure, e.g., cross subsidies to improve public transit. Furthermore, the direct lump-sum transfers of the collected revenues to users studied in this work is a specific instance of differential congestion pricing schemes wherein the price on a given path may differ by user group. As a result, it would be worthwhile to investigate a broader class of such group-specific differential congestion pricing mechanisms. Finally, it would also be valuable to explore more complex behavior such as inter-group coalitions~\cite{ANDELMAN2009289}.

\fi

\section*{Acknowledgment}

The authors thank Alon Eden and Kira Goldner for insightful discussions on non-atomic congestion games.

\bibliographystyle{unsrt}
\bibliography{main}

\begin{thebibliography}{10}

\bibitem{pigou}
A.~C. Pigou.
\newblock {\em The economics of welfare}.
\newblock Macmillan London, 1920.

\bibitem{tsekeris2009design}
Theodore Tsekeris and Stefan Vo{\ss}.
\newblock Design and evaluation of road pricing: state-of-the-art and
  methodological advances.
\newblock {\em NETNOMICS: Economic Research and Electronic Networking},
  10(1):5--52, 2009.

\bibitem{how-bad-is-selfish}
T.~Roughgarden and \'{E}. Tardos.
\newblock How bad is selfish routing?
\newblock {\em J. ACM}, 49(2):236–259, 2002.

\bibitem{roughgarden2005selfish}
Tim Roughgarden.
\newblock {\em Selfish routing and the price of anarchy}.
\newblock MIT press, 2005.

\bibitem{Sheffi1985}
Y.~Sheffi.
\newblock {\em Urban Transportation Networks: Equilibrium Analysis with
  Mathematical Programming Methods}.
\newblock Prentice-Hall, 1985.

\bibitem{SMALL2001310}
Kenneth~A. Small and Jia Yan.
\newblock The value of “value pricing” of roads: Second-best pricing and
  product differentiation.
\newblock {\em Journal of Urban Economics}, 49(2):310 -- 336, 2001.

\bibitem{eliasson2001road}
Jonas Eliasson.
\newblock Road pricing with limited information and heterogeneous users: A
  successful case.
\newblock {\em The annals of regional science}, 35(4):595--604, 2001.

\bibitem{manville-empirical}
Michael Manville and Emily Goldman.
\newblock Would congestion pricing harm the poor? do free roads help the poor?
\newblock {\em Journal of Planning Education and Research}, 38(3):329--344,
  2018.

\bibitem{gemici_et_al:LIPIcs:2019:10270}
Kurtulus Gemici, Elias Koutsoupias, Barnab{\'e} Monnot, Christos~H.
  Papadimitriou, and Georgios Piliouras.
\newblock {Wealth Inequality and the Price of Anarchy}.
\newblock In {\em International Symposium on Theoretical Aspects of Computer
  Science}, volume 126 of {\em Leibniz International Proceedings in
  Informatics}, pages 31:1--31:16. Schloss Dagstuhl--Leibniz-Zentrum fuer
  Informatik, 2019.

\bibitem{nyt-cp}
Azi Paybarah.
\newblock Congestion pricing: Mass transit savior or tax on the working class?
\newblock New York Times, March 2019.

\bibitem{CP-Low-acceptance}
S.~Jaensirisak, M.~Wardman, and A.~D. May.
\newblock Explaining variations in public acceptability of road pricing
  schemes.
\newblock {\em Journal of Transport Economics and Policy}, 39(2):127--153,
  2005.

\bibitem{wachs2005then}
Martin Wachs.
\newblock Then and now: The evolution of congestion pricing in transportation
  and where we stand today.
\newblock In {\em Transportation Research Board Conference Proceedings}, 2005.

\bibitem{WU20121273}
Di~Wu, Yafeng Yin, Siriphong Lawphongpanich, and Hai Yang.
\newblock Design of more equitable congestion pricing and tradable credit
  schemes for multimodal transportation networks.
\newblock {\em Transportation Research Part B: Methodological},
  46(9):1273--1287, 2012.

\bibitem{fair-deSouza}
Patrick~De Corla-Souza.
\newblock Fair highway networks: A new approach to eliminate congestion on
  metropolitan freeways.
\newblock {\em Public Works Management \& Policy}, 9(3):196--205, 2005.

\bibitem{GUO2010972}
Xiaolei Guo and Hai Yang.
\newblock Pareto-improving congestion pricing and revenue refunding with
  multiple user classes.
\newblock {\em Transportation Research Part B: Methodological}, 44(8):972--982,
  2010.

\bibitem{heterogeneous-pricing-roughgarden}
Richard Cole, Yevgeniy Dodis, and Tim Roughgarden.
\newblock Pricing network edges for heterogeneous selfish users.
\newblock In {\em Symposium on Theory of Computing}, page 521–530.
  Association for Computing Machinery, 2003.

\bibitem{multicommodity-extension}
L.~{Fleischer}, K.~{Jain}, and M.~{Mahdian}.
\newblock Tolls for heterogeneous selfish users in multicommodity networks and
  generalized congestion games.
\newblock In {\em Symposium on Foundations of Computer Science}, pages
  277--285. IEEE, 2004.

\bibitem{pan1993h}
Zigang Pan and Tamer Ba{$\c{s}$}ar.
\newblock H$\infty$-optimal control for singularly perturbed systems. part i:
  Perfect state measurements.
\newblock {\em Automatica}, 29(2):401--423, 1993.

\bibitem{sofronova2020traffic}
Elena~A Sofronova and Askhat~I Diveev.
\newblock Traffic flows optimal control problem with full information.
\newblock In {\em 2020 IEEE 23rd International Conference on Intelligent
  Transportation Systems (ITSC)}, pages 1--6. IEEE, 2020.

\bibitem{sabag2021regret}
Oron Sabag, Gautam Goel, Sahin Lale, and Babak Hassibi.
\newblock Regret-optimal controller for the full-information problem.
\newblock In {\em 2021 American Control Conference (ACC)}, pages 4777--4782.
  IEEE, 2021.

\bibitem{arora2020private}
Kashish Arora, Fanyin Zheng, and Karan Girotra.
\newblock Private vs. pooled transportation: Customer preference and congestion
  management.
\newblock In {\em INFORMS International Conference on Service Science}, pages
  59--74. Springer, 2020.

\bibitem{Buchholzrdab050}
Nicholas Buchholz.
\newblock {Spatial Equilibrium, Search Frictions, and Dynamic Efficiency in the
  Taxi Industry}.
\newblock {\em The Review of Economic Studies}, 89(2):556--591, 09 2021.

\bibitem{cohen2016using}
Peter Cohen, Robert Hahn, Jonathan Hall, Steven Levitt, and Robert Metcalfe.
\newblock Using big data to estimate consumer surplus: The case of uber.
\newblock Technical report, National Bureau of Economic Research, 2016.

\bibitem{jalota-acm-eaamo}
Devansh Jalota, Kiril Solovey, Karthik Gopalakrishnan, Stephen Zoepf, Hamsa
  Balakrishnan, and Marco Pavone.
\newblock When efficiency meets equity in congestion pricing and revenue
  refunding schemes.
\newblock In {\em Equity and Access in Algorithms, Mechanisms, and
  Optimization}, New York, NY, USA, 2021. Association for Computing Machinery.

\bibitem{Bertsimas-price-of-fairness}
D.~Bertsimas, Vivek Farias, and N.~Trichakis.
\newblock The price of fairness.
\newblock {\em Operations Research}, 59(1):17--31, 2011.

\bibitem{dwork-fairness-through-awareness}
Cynthia Dwork, Moritz Hardt, Toniann Pitassi, Omer Reingold, and Richard Zemel.
\newblock Fairness through awareness.
\newblock In {\em Innovations in Theoretical Computer Science Conference}, page
  214–226, New York, NY, USA, 2012. Association for Computing Machinery.

\bibitem{so-routing-seminal}
O.~Jahn, R.~Möhring, A.~Schulz, and N.~Stier-Moses.
\newblock System-optimal routing of traffic flows with user constraints in
  networks with congestion.
\newblock {\em Operations Research}, 53(4):600--616, 2005.

\bibitem{Roughgarden2002HowUI}
T.~Roughgarden.
\newblock How unfair is optimal routing?
\newblock In {\em Symposium on Discrete Algorithms}, pages 203--204.
  {ACM/SIAM}, 2002.

\bibitem{ANGELELLI20161}
E.~Angelelli, I.~Arsik, V.~Morandi, M.~Savelsbergh, and M.~Speranza.
\newblock Proactive route guidance to avoid congestion.
\newblock {\em Transportation Research Part B: Methodological}, 94:1--21, 2016.

\bibitem{ANGELELLI2018234}
E.~Angelelli, V.~Morandi, and M.~Speranza.
\newblock Congestion avoiding heuristic path generation for the proactive route
  guidance.
\newblock {\em Computers \& Operations Research}, 99:234--248, 2018.

\bibitem{ANGELELLI2020}
E.~Angelelli, V.~Morandi, M.~Savelsbergh, and M.~Speranza.
\newblock System optimal routing of traffic flows with user constraints using
  linear programming.
\newblock {\em European Journal of Operational Research}, 2020.

\bibitem{Jalota.ea.CDC21.extended}
Devansh Jalota, Kiril Solovey, Stephen Zoepf, and Marco Pavone.
\newblock Balancing fairness and efficiency in traffic routing via interpolated
  traffic assignment.
\newblock {\em CoRR}, abs/2104.00098, 2021.

\bibitem{weitzman-seminal}
Martin~L. Weitzman.
\newblock Is the price system or rationing more efficient in getting a
  commodity to those who need it most?
\newblock {\em Bell Journal of Economics}, 8(2):517--524, 1977.

\bibitem{PPP-Besley-coate}
Timothy Besley and Stephen Coate.
\newblock Public provision of private goods and the redistribution of income.
\newblock {\em The American Economic Review}, 81(4):979--984, 1991.

\bibitem{CONDORELLI2013582}
Daniele Condorelli.
\newblock Market and non-market mechanisms for the optimal allocation of scarce
  resources.
\newblock {\em Games and Economic Behavior}, 82:582--591, 2013.

\bibitem{RAM-Akbarpour}
Mohammad Akbarpour, Piotr Dworczak, and Scott~Duke Kominers.
\newblock {Redistributive allocation mechanisms}.
\newblock GRAPE Working Papers~40, GRAPE Group for Research in Applied
  Economics, 2020.

\bibitem{small1992using}
Kenneth~A Small.
\newblock Using the revenues from congestion pricing.
\newblock {\em Transportation}, 19(4):359--381, 1992.

\bibitem{goodwin1989rule}
Phil~B Goodwin.
\newblock The rule of three: a possible solution to the political problem of
  competing objectives for road pricing.
\newblock {\em Traffic engineering \& control}, 30(10):495--497, 1989.

\bibitem{vickrey1969congestion}
William~S Vickrey.
\newblock Congestion theory and transport investment.
\newblock {\em The American Economic Review}, 59(2):251--260, 1969.

\bibitem{arnott1994}
Richard Arnott, André de~Palma, and Robin Lindsey.
\newblock The welfare effects of congestion tolls with heterogeneous commuters.
\newblock {\em Journal of Transport Economics and Policy}, 28(2):139--161,
  1994.

\bibitem{DAGANZO1995139}
Carlos~F. Daganzo.
\newblock A pareto optimum congestion reduction scheme.
\newblock {\em Transportation Research Part B: Methodological}, 29(2):139--154,
  1995.

\bibitem{ADLER2001447}
Jeffrey~L. Adler and Mecit Cetin.
\newblock A direct redistribution model of congestion pricing.
\newblock {\em Transportation Research Part B: Methodological}, 35(5):447--460,
  2001.

\bibitem{YANG20041}
Hai Yang and Hai-Jun Huang.
\newblock The multi-class, multi-criteria traffic network equilibrium and
  systems optimum problem.
\newblock {\em Transportation Research Part B: Methodological}, 38(1):1--15,
  2004.

\bibitem{dabla2015causes}
Era Dabla-Norris, Kalpana Kochhar, Nujin Suphaphiphat, Frantisek Ricka, and
  Evridiki Tsounta.
\newblock {\em Causes and consequences of income inequality: A global
  perspective}.
\newblock International Monetary Fund, 2015.

\bibitem{Patriksson15}
Michael Patriksson.
\newblock {\em The Traffic Assignment Problem: {M}odels and Methods}.
\newblock Dover Publications, 2015.

\bibitem{lawphongpanich2010solving}
S~Lawphongpanich and Y~Yin.
\newblock Solving pareto-improving congestion pricing for general road
  networks.
\newblock {\em Transportation Research Part C}, 18(2):234--246, 2010.

\bibitem{METRP}
H.~Yang and H.~J. Huang.
\newblock {\em Mathematical and Economic Theory of Road Pricing}.
\newblock Emerald Publishing, 2005.

\bibitem{hearn1998solving}
Donald~W Hearn and Motakuri~V Ramana.
\newblock Solving congestion toll pricing models.
\newblock In {\em Equilibrium and advanced transportation modelling}, pages
  109--124. Springer, 1998.

\bibitem{bisschop2016aimms}
Johannes Bisschop.
\newblock {\em {AIMMS} optimization modeling}.
\newblock lulu.com, 2016.

\bibitem{rawls2020theory}
John Rawls.
\newblock {\em A theory of justice}.
\newblock Harvard university press, 2020.

\bibitem{ANDELMAN2009289}
Nir Andelman, Michal Feldman, and Yishay Mansour.
\newblock Strong price of anarchy.
\newblock {\em Games and Economic Behavior}, 65(2):289--317, 2009.

\bibitem{hall_promoting_2006}
Ralph~P. Hall and Joseph~M. Sussman.
\newblock Promoting the concept of sustainable transportation within the
  federal system - the need to reinvent the {U.S.}, 2006.

\bibitem{daganzo1994cell}
Carlos~F Daganzo.
\newblock The cell transmission model: {A} dynamic representation of highway
  traffic consistent with the hydrodynamic theory.
\newblock {\em Transportation Research Part {B}: {M}ethodological},
  28(4):269--287, 1994.

\bibitem{gartner1980optimal}
Nathan~H Gartner.
\newblock Optimal traffic assignment with elastic demands: A review part {I}.
  analysis framework.
\newblock {\em Transportation Science}, 14(2):174--191, 1980.

\bibitem{orpn-2023}
Devansh Jalota, Dario Paccagnan, Maximilian Schiffer, and Marco Pavone.
\newblock Online routing over parallel networks: Deterministic limits and
  data-driven enhancements.
\newblock {\em INFORMS Journal on Computing}, 0(0):null, 0.

\end{thebibliography}

\ifarxiv
\section{Proofs}
\fi

\ifarxiv
\subsection{Constant Income Transfer Property} \label{sec:const-inc-transfer}
In this section, we show that the constant income transfer property follows directly from the regressive and progressive tax properties of the wealth inequality measure $W$, as claimed in Section~\ref{sec:model}. In particular, we show that if the initial income distribution is $\q$ and each person is transferred a non-positive amount of money $-\lambda$, where $0 \leq \lambda < \min_{g \in \G} q_g$, then the wealth inequality cannot decrease, i.e.,  $W(\mathbf{q} - \lambda \mathbf{1}) \geq W(\mathbf{q})$.

We note that at the new income distribution $\Bar{\q} = \mathbf{q} - \lambda \mathbf{1}$, each user in group $g$ has the following income:
\begin{align*}
    \Bar{q}_g &= q_g - \lambda = q_g \left(1 - \frac{\lambda}{q_g} \right).
\end{align*}
Note that if $q_g \leq q_{g'}$ for any two groups $g, g'$, then $1 - \frac{\lambda}{q_g} \leq 1 - \frac{\lambda}{q_{g'}}$. Thus, by the regressive tax property, we observe that $W(\mathbf{q} - \lambda \mathbf{1}) \geq W(\mathbf{q})$. We finally note that the claim that $W(\mathbf{q} + \lambda \mathbf{1}) \leq W(\mathbf{q})$ for any $0 \leq \lambda < \min_{g \in \G} q_g$ follows by a similar analysis wherein we use the progressive tax property. This proves our claim that the wealth inequality measure $W$ satisfies the constant income transfer property.

\fi

\ifarxiv
\subsection{Proof of Equation~\ref{eq:total-cost-relation}} \label{apdx:kkt-multiclass-ue}

In this section, we use the first order necessary and sufficient KKT conditions of the well studied multi-class user equilibrium optimization problem~\citep{YANG20041}
\[\f = \argmin_{\f' \in \Omega} \sum_{e \in E} \int_{0}^{x'_{e}} t_{e}(\omega) d \omega+\sum_{e \in E} \sum_{g \in \G} \frac{1}{v_{g}} {x'}_{e}^{g} \tau_{e},\]
to prove that the following holds:
\begin{align}
    C_{\ttau} = \sum_{g \in \G} \mu^{g}(\ttau, \0) d_g - \sum_{e \in E} \tau_e x_e.
\end{align} 
Here $\ttau$ is congestion-pricing scheme and $\f$ is an exogenous $\ttau$-equilibrium with edge flow representation $\x$. Note that the edge flows $\x$ are unique by the strict convexity of the travel time function.

The following exogenous-equilibrium conditions follow directly from the KKT conditions of the above optimization problem:
\begin{align*}
    \sum_{e \in P} \left(v_g t_e(x_e) + \tau_{e} \right) = \mu^{g}(\ttau, \0), \quad \text{ if } f_{P,g}>0, P \in \P_g, g \in \G, \\
    \sum_{e \in P} \left(v_g t_e(x_e) + \tau_{e} \right) \geq \mu^{g}(\ttau, \0), \quad \text{ if } f_{P,g}=0, P \in \P_g, g \in \G.
\end{align*}
From the above equilibrium conditions and the fact that $\sum_{P \in \P_g} f_{P, g} = d_g$, we obtain that:
\begin{align*}
    \sum_{g \in \G} \mu^{g}(\ttau, \0) d_g 
    &= \sum_{g \in \G} \sum_{P \in \P_g} f_{P, g} \mu^{g}(\ttau, \0) = \sum_{g \in \G} \sum_{P \in \P_g} f_{P, g} \sum_{e \in P} \left(v_g t_e(x_e) + \tau_{e} \right), \\
    &= \sum_{g \in \G} \sum_{P \in \P_g} f_{P, g} \sum_{e \in E} \left(v_g t_e(x_e) + \tau_{e} \right) \delta_{e, P} = \sum_{e \in E} \sum_{g \in \G} \sum_{P \in \P_g} f_{P, g} \left(v_g t_e(x_e) + \tau_{e} \right) \delta_{e, P}, \\
    &= \sum_{e \in E} \sum_{g \in \G} \sum_{P \in \P_g: e \in P} f_{P, g} \left(v_g t_e(x_e) + \tau_{e} \right) = \sum_{e \in E} \sum_{g \in \G} \left(v_g t_e(x_e) + \tau_{e} \right) \sum_{P \in \P_g: e \in P} f_{P, g},  \\
    &= \sum_{e \in E} \sum_{g \in \G} x_e^{g} \left(v_g t_e(x_e) + \tau_{e} \right) = \sum_{e \in E} \sum_{g \in \G} x_e^{g} v_g t_e(x_e) + \sum_{e \in E} x_e \tau_e
\end{align*}
where $\delta_{e, P} = 1$ if edge $e \in P$ and otherwise it is 0. Note that the above analysis implies Equation~\eqref{eq:total-cost-relation} since $C_{\ttau} = \sum_{e \in E} \sum_{g \in \G} x_e^{g} v_g t_e(x_e) = \sum_{g \in \G} \mu^{g}(\ttau, \0) d_g - \sum_{e \in E} x_e \tau_e$. This proves our claim.

\begin{remark}
We note that since the total tolls collected and user travel costs $\mu^{g}(\ttau, \0)$ are unique at any equilibrium flow~\citep{GUO2010972}, the total travel cost $C_{\ttau}$ is also unique for any equilibrium induced by the edge tolls $\ttau$. Furthermore, the ex-post income of each user group $g$ is also the same under any equilibrium induced by the edge tolls $\ttau$ since the user travel cost $\mu^{g}(\ttau, \0)$ is unique at any equilibrium flow~\citep{GUO2010972}.
\end{remark}

\fi

\ifarxiv
\subsection{Proof of Lemma~\ref{lem:PI-CPRR}} \label{apdx:lemma-pi-cprr-pf}
Lemma~\ref{lem:PI-CPRR} is a slight variation on a previous result by \cite{GUO2010972}. We provide a proof for this modified lemma for the sake of completeness.

To prove the lemma, we first show that $\sum_{g \in \G} r_g d_g = \sum_{e \in E} \tau_e x_e$, where $x_e$ is the flow on edge $e$ corresponding to the edge decomposition of any exogenous $\ttau$-equilibrium flow $\f$. Indeed, it follows that 
\begin{align*}
    \sum_{g \in \G} r_g d_g &= \sum_{g \in \G} \left[ \mu^g(\ttau,\0) - \mu^g(\0,\0) + \frac{\alpha_g}{d_g}(C_{\0} - C_{\ttau}) \right] d_g = \sum_{g \in \G} \mu^g(\ttau,\0) d_g - \sum_{g \in \G} \mu^g(\0,\0) d_g + \sum_{g \in \G} \alpha_g (C_{\0} - C_{\ttau}) \\
    &= C_{\ttau} + \sum_{e \in E} \tau_e x_e - C_{\0} + C_{\0} - C_{\ttau} = \sum_{e \in E} \tau_e x_e,
\end{align*}
where we have used Equation~\eqref{eq:total-cost-relation} to obtain that $C_{\ttau} = \sum_{g \in \G} \mu^g(\f,\ttau,\0) d_g - \sum_{e \in E} \tau_e x_e$. Furthermore, we leveraged Equation~\eqref{eq:total-cost-relation} for the untolled outcome to obtain that $C_{\0} = \sum_{g \in \G} \mu^g(\0,\0) d_g$.

Next, we show that this CPRR scheme is user-favorable, which follows from the following inequalities:
\begin{align*}
    \mu^g(\ttau,\r) &= \mu^g(\ttau,\0) - r_g = \mu^g(\ttau,\0) - \left( \mu^g(\ttau,\0) - \mu^g(\0,\0) + \frac{\alpha_g}{d_g}(C_{\0} - C_{\ttau}) \right) \\
    &= \mu^g(\0,\0) - \alpha_g(C_{\0} - C_{\ttau}) \leq \mu^g(\0,\0).
\end{align*}

\fi

\ifarxiv
\subsection{Discrete Gini Coefficient Satisfies the Regressive Income Tax Property} \label{apdx:gini-regressive}

We show that for any two income profiles $\mathbf{q}$ and $\Tilde{\q}$ with $\Tilde{q}_g = \delta_g q_g$, where $0<\delta_g \leq \delta_{g'}$ if $q_g \leq q_{g'}$ for any two groups $g, g'$, then $W(\Tilde{\q}) \geq W(\mathbf{q})$ for the discrete Gini coefficient wealth inequality measure.

To prove this, we first note by the scale independence property that, we can restrict our attention to scaling factors $\delta_g$ for each $g$ that leave the mean income of all users unchanged, i.e., $\Delta(\Tilde{\q}) = \Delta(\q)$. This is because, if the mean income  of the new income profile $\Tilde{\q}$ is different than that of $\mathbf{q}$ then we can multiply the new incomes with a scaling constant $\lambda>0$ to ensure that the mean income of $\lambda \Tilde{\q}$ is exactly that of $\mathbf{q}$. Then, by scale independence we have that $W(\lambda \Tilde{\q}) = W(\Tilde{\q})$, and so, without loss of generality, we focus on the set of scaling factors $\delta_g$ for each $g$, such that $\Delta(\Tilde{\q}) = \Delta(\q)$.

Next, we note the following:
\begin{align*}
    W(\mathbf{q}) &= \frac{1}{2 \left(\sum_{g \in \G} d_g \right)^2 \Delta(\q)} \sum_{g_1, g_2 \in \G} d_{g_1} d_{g_2} |q_{g_1} - q_{g_2}| = \frac{1}{2 \left(\sum_{g \in \G} d_g \right)^2 \Delta(\Tilde{\q})} \sum_{g_1, g_2 \in \G} d_{g_1} d_{g_2} |q_{g_1} - q_{g_2}| \\
    & \leq \frac{1}{2 \left(\sum_{g \in \G} d_g \right)^2 \Delta(\Tilde{\q})} \sum_{g_1, g_2 \in \G} d_{g_1} d_{g_2} |\delta_{g_1} q_{g_1} - \delta_{g_2} q_{g_2}| = \frac{1}{2 \left(\sum_{g \in \G} d_g \right)^2 \Delta(\Tilde{\q})} \sum_{g_1, g_2 \in \G} d_{g_1} d_{g_2} |\Tilde{q}_{g_1} - \Tilde{q}_{g_2}| \\
    &= W(\Tilde{\q}),
\end{align*}
where the second equality follows since $\Delta(\Tilde{\q}) = \Delta(\q)$, and the third inequality follows since $|q_{g_1} - q_{g_2}| \leq |\delta_{g_1} q_{g_1} - \delta_{g_2} q_{g_2}|$ for any $g_1, g_2$ as $0<\delta_g \leq \delta_{g'}$ if $q_g \leq q_{g'}$ for any two groups $g, g'$.

Thus, we have shown that the discrete Gini coefficient wealth inequality measure satisfies the regressive tax property.

\fi
\ifarxiv

\subsection{Proof of Proposition~\ref{prop:opt-rev-refunds-discrete-gini}} \label{apdx:opt-rev-discrete-gini}

We first note that since $\frac{1}{2 \left(\sum_{g \in \G} d_{g} \right)^2 \Delta(\q(\0, \0)  + \cc)}$ is a constant the optimal solution of the Problem~\eqref{eq:Gini-obj}-\eqref{eq:pareto-imp-refund} is equal to the optimal solution of the following problem
\begin{mini!}|s|[2]                   
    {c_{g}}                              
    { \sum_{g_1, g_2 \in \G} d_{g_1} d_{g_2} |\Tilde{q}_{g_1} + c_{g_1} - \Tilde{q}_{g_2} - c_{g_2}|, \label{eq:Gini-obj-simple}}   
    {\label{eq:Eg001-simple}}             
    {}                                
    \addConstraint{\sum_{g \in \G} c_{g} d_{g}}{= C_{\0} - C_{\ttau}, \label{eq:Gini-con1-simple}}    
    \addConstraint{c_{g}}{ \geq 0, \quad \forall g \in \G, \label{eq:pareto-imp-refund-simple}}
\end{mini!}
where $\Tilde{q}_g = q_{g}(\0, \0)$ for conciseness. We further note that the above problem can be reformulated as the following linear program, where we add the variables $y_{g_1,g_2}$ to correspond to the absolute value term in the objective:
\begin{mini!}|s|[2]                   
    {c_{g}, y_{g_1, g_2}}                              
    { \sum_{g_1, g_2 \in \G} d_{g_1} d_{g_2} y_{g_1, g_2}, \label{eq:Gini-obj-simpleLP}}   
    {\label{eq:Eg001-simpleLP}}             
    {}                                
    \addConstraint{\sum_{g \in \G} c_{g} d_{g}}{= C_{\0} - C_{\ttau}, \label{eq:Gini-con1-simpleLP}}    
    \addConstraint{c_{g}}{ \geq 0, \quad \forall g \in \G, \label{eq:pareto-imp-refund-simpleLP}}
    \addConstraint{y_{g_1,g_2}}{ \geq \Tilde{q}_{g_1} + c_{g_1} - \Tilde{q}_{g_2} - c_{g_2}, \quad \forall g_1 \in \G, g_2 \in \G, \label{eq:absLP1}}
    \addConstraint{y_{g_1,g_2}}{ \geq -(\Tilde{q}_{g_1} + c_{g_1} - \Tilde{q}_{g_2} - c_{g_2}), \quad \forall g_1 \in \G, g_2 \in \G, \label{eq:absLP2}}
\end{mini!}

We now compute the optimal solution of Problem~\eqref{eq:Gini-con1-simpleLP}-\eqref{eq:absLP2} by deriving the first order conditions of the optimization problem. To this end, let $\lambda$ be the dual variable of the Constraint~\eqref{eq:Gini-con1-simpleLP}, $l_g$ be the dual variable of the Constraint~\eqref{eq:pareto-imp-refund-simpleLP} for each group $g$, $\eta_{g_1, g_2}$ be the dual variable of Constraint~\eqref{eq:absLP1}, and $\xi_{g_1,g_2}$ be the dual variable of Constraint~\eqref{eq:absLP2}. Then, we can formulate the following Lagrangian:
\begin{align*}
    \mathcal{L} = &\sum_{g_1, g_2 \in \G} d_{g_1} d_{g_2} y_{g_1, g_2} - \lambda \left( \sum_{g \in \G} c_{g} d_{g} - (C_{\0} - C_{\ttau}) \right) - \sum_{g \in \G} l_g c_g - \sum_{g_1, g_2 \in \G} \eta_{g_1,g_2} (y_{g_1,g_2} - (\Tilde{q}_{g_1} + c_{g_1} - \Tilde{q}_{g_2} - c_{g_2})) \\ &- \sum_{g_1, g_2 \in \G} \xi_{g_1,g_2} (y_{g_1,g_2} + \Tilde{q}_{g_1} + c_{g_1} - \Tilde{q}_{g_2} - c_{g_2})
\end{align*}
Taking the first order derivative condition of this Lagrangian with respect to $y_{g_1,g_2}$, we obtain that
\begin{align}
    &\frac{\partial \mathcal{L}}{\partial y_{g_1,g_2}} = d_{g_1} d_{g_2} - \eta_{g_1,g_2} - \xi_{g_1,g_2} = 0, \nonumber \\
    &\implies \eta_{g_1,g_2} + \xi_{g_1,g_2} = d_{g_1} d_{g_2}, \quad \forall g_1 \in \G, g_2 \in \G. \label{eq:mainEqualityRelation}
\end{align}
Next, taking the first order derivative condition of this Lagrangian with respect to $c_g$, we obtain that
\begin{align*}
    \frac{\partial \mathcal{L}}{\partial c_g} = - \lambda d_g - l_g + \sum_{g_2 \neq g} \eta_{g, g_2} - \sum_{g_2 \neq g} \eta_{g_2, g} + \sum_{g_2 \neq g} \xi_{g_2, g} - \sum_{g_2 \neq g} \xi_{g, g_2} = 0.
\end{align*}
Using this relationship, the sign constraint on the dual variable $l_g$, i.e., $l_g \geq 0$, and the complimentary slackness relation that $l_g c_g = 0$ for all $g$, we obtain the following first order conditions:
\begin{align*}
    &\sum_{g_2 \neq g} (\xi_{g_2, g} - \eta_{g_2, g}) + \sum_{g_2 \neq g} (\eta_{g, g_2} - \xi_{g, g_2}) \geq \lambda d_g, \quad \forall g \in \G, \\
    &\sum_{g_2 \neq g} (\xi_{g_2, g} - \eta_{g_2, g}) + \sum_{g_2 \neq g} (\eta_{g, g_2} - \xi_{g, g_2}) = \lambda d_g, \quad \forall g \in \G, \text{ s.t. } c_g > 0. 
\end{align*}
Since the dual variables $\xi_{g_2, g}$ and $\eta_{g, g_2}$ correspond to identical constraints it follows by symmetry that $\xi_{g_2, g} = \eta_{g, g_2}$ (and analogously that $\eta_{g_2, g} = \xi_{g, g_2}$). As a result, the above first order conditions can be simplified as:
\begin{align}
    &2 \sum_{g_2 \neq g} (\eta_{g, g_2} - \xi_{g, g_2}) \geq \lambda d_g, \quad \forall g \in \G, \label{eq:helper1_} \\
    &2\sum_{g_2 \neq g} (\eta_{g, g_2} - \xi_{g, g_2}) = \lambda d_g, \quad \forall g \in \G, \text{ s.t. } c_g > 0.  \label{eq:helper2_}
\end{align}
Next, to express the term $\eta_{g, g_2} - \xi_{g, g_2}$ in terms of the demands $d_g$ we use Equation~\eqref{eq:mainEqualityRelation} to write $\eta_{g, g_2} - \xi_{g, g_2}$. To this end, first note that if $y_{g,g_2} = \Tilde{q}_{g} + c_{g} - \Tilde{q}_{g_2} - c_{g_2} = -(\Tilde{q}_{g} + c_{g} - \Tilde{q}_{g_2} - c_{g_2})$, then $\xi_{g,g_2} = \eta_{g,g_2}$ as the two equations are identical, and thus $\xi_{g,g_2} - \eta_{g,g_2} = 0$. 

As a result, for the rest of this proof we consider the case when $\Tilde{q}_{g} + c_{g} - \Tilde{q}_{g_2} - c_{g_2} \neq -(\Tilde{q}_{g} + c_{g} - \Tilde{q}_{g_2} - c_{g_2})$, i.e., either (i) $y_{g,g_2} = \Tilde{q}_{g} + c_{g} - \Tilde{q}_{g_2} - c_{g_2}$, or (ii) $y_{g,g_2} = -(\Tilde{q}_{g} + c_{g} - \Tilde{q}_{g_2} - c_{g_2})$. Note in case (i) that $\xi_{g,g_2} \geq 0$ and $\eta_{g,g_2} = 0$ by the complimentary slackness relation, and thus it holds that $\xi_{g,g_2} = d_{g,g_2}$ by Equation~\eqref{eq:mainEqualityRelation}. Analogously, it holds in case (ii) that $\eta_{g,g_2} = d_{g,g_2}$.

Furthermore, observe that $y_{g,g_2} = \Tilde{q}_{g} + c_{g} - \Tilde{q}_{g_2} - c_{g_2}$ only when $\Tilde{q}_{g} + c_{g} - \Tilde{q}_{g_2} - c_{g_2} > -(\Tilde{q}_{g} + c_{g} - \Tilde{q}_{g_2} - c_{g_2})$, i.e., $\Tilde{q}_{g_2} + c_{g_2} > \Tilde{q}_{g} + c_{g}$, and similarly $y_{g,g_2} = -(\Tilde{q}_{g} + c_{g} - \Tilde{q}_{g_2} - c_{g_2})$ only when $\Tilde{q}_{g_2} + c_{g_2} < \Tilde{q}_{g} + c_{g}$. Combining the above derived relations, we obtain that
\begin{align}
    2 \sum_{g_2 \neq g} (\eta_{g, g_2} - \xi_{g, g_2}) = 2\sum_{g_2: q_{g_2} + r_{g_2}< q_{g} + r_{g}} d_{g_2} d_{g} - 2 \sum_{g_2: q_{g_2} + r_{g_2} > q_{g} + r_{g}} d_{g_2} d_{g}. \label{eq:newHelper_}
\end{align}
Using Equation~\eqref{eq:newHelper_} and the first order conditions derived in Equations~\eqref{eq:helper1_} and~\eqref{eq:helper2_}, we obtain that
\begin{align*}
    &2 \sum_{g_2: q_{g_2} + r_{g_2}< q_{g} + r_{g}} d_{g_2} - 2 \sum_{g_2: q_{g_2} + r_{g_2} > q_{g} + r_{g}} d_{g_2} \geq \lambda, \quad \forall g \in \G, \\
    &2 \sum_{g_2: q_{g_2} + r_{g_2}< q_{g} + r_{g}} d_{g_2} - 2 \sum_{g_2: q_{g_2} + r_{g_2} > q_{g} + r_{g}} d_{g_2} = \lambda, \quad \forall g \in \G, \text{ s.t. } c_{g} > 0,
\end{align*}
where we divided both sides of the equation by $d_g$. From these equations, we observe that the income group(s) that receive strictly positive transfers $c_g>0$ are those for whom the above equation is met with an equality. Since $\lambda$ is a fixed quantity, it follows that the above equation is met with equality for groups $g$ with the minimum value of the following term:
$$2 \sum_{g_2: q_{g_2} + r_{g_2}< q_{g} + r_{g}} d_{g_2} - 2 \sum_{g_2: q_{g_2} + r_{g_2} > q_{g} + r_{g}} d_{g_2}$$
Note that this term is the smallest only for the lowest income groups after their revenue refunds. That is, all groups that receive a refund that results in a strictly positive transfer $c_g>0$ have exactly the same income. 

The above observation implies that one way to achieve the optimal solution of Problem~\eqref{eq:Gini-obj-simpleLP}-\eqref{eq:absLP2} is to provide positive transfers to the lowest income users until their incomes equalize with the second lowest income group. Then both these groups can be given revenue until their income rises to the third lowest income group and so on, as described in Algorithm~\ref{alg:greedy}. This process can be repeated until the total pool of transfers $C_{\0} - C_{\ttau}$ is completely exhausted.

\fi

\ifarxiv
\subsection{Proof of Proposition~\ref{prop:two-link-eq-countereg}} \label{eq:pfprop4}

We formally define the instance that is described in \ifarxiv Figure~\ref{fig:prop-4-countereg}\else Fig.~\ref{fig:prop-4-countereg}\fi. Consider a two edge parallel network, having one origin and one destination, with travel time functions $t_1(x_1) = 2x_1$ and $t_2(x_2) = 4+x_2$. Consider three user classes $H, M, L$ representing high, medium, and low incomes, respectively. Let the demands of the three classes be $d_H = 2, d_M = 1, d_L = 5$, where the incomes are $q_H = 2q_M$, $q_M$ and $q_L$, where $q_M (1-0.008) = q_L (1-0.010) + \frac{0.014q_M}{5}$, and the relative importance of the congestion game is given by a factor $\beta = 1$.  Further, let the values-of-time of the users be scaled proportions of their income by a factor of $0.001$, i.e., $v_H = 0.002q_M, v_M = 0.001q_M, v_L = 0.001q_L$. 

We define the following congestion pricing: $\tau_1 = 0.008q^M$, and $\tau_2=0$, i.e., edge $2$ is untolled. Given this pricing $\ttau$, we define the refunding policy $\r$, which is known to all users, and is derived from Algorithm~\ref{alg:greedy}. That is, we first provide enough refunds to ensure that all groups exactly meet their untolled user equilibrium costs, and then give any remaining refunds to the lowest income group (until their income equals that of the second lowest income group and so on).

\begin{figure}[!ht]
      \centering
      \includegraphics[width=0.5\linewidth]{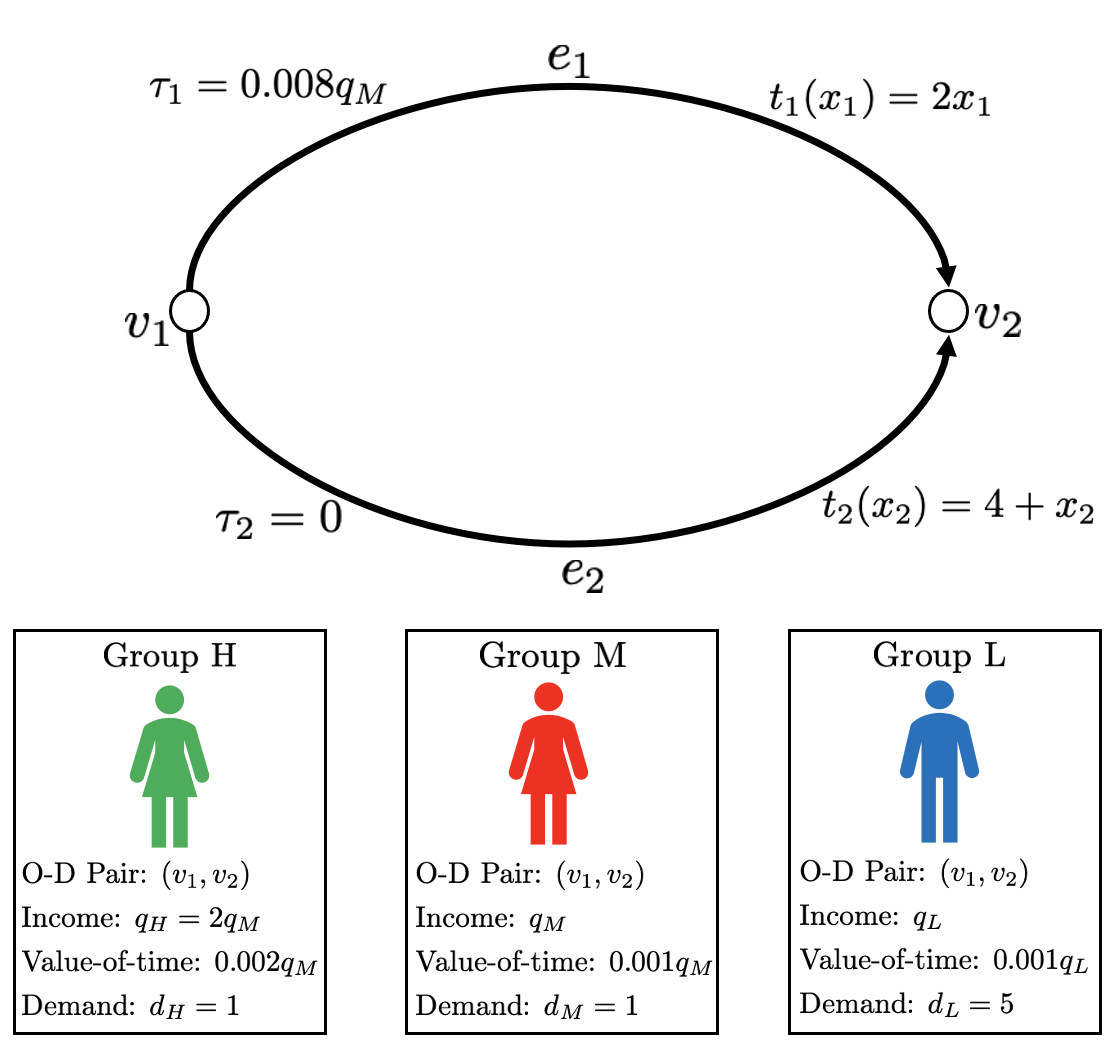}
      \caption{{\small \sf A two-edge parallel network and three user-group instance for which the exogenous $\ttau$-equilibrium is not an endogenous $(\ttau, \r)$-equilibrium with coalitions, under the revenue refunding scheme $\r$ resulting from Algorithm~\ref{alg:greedy}. For $\ttau = (0.008 q_M, 0)$, only users in the high income group $H$ use edge $e_1$ at the exogenous $\ttau$-equilibrium. However, this is not an endogenous $(\ttau, \r)$-equilibrium since all the users in group $M$ can deviate to use edge $e_1$, which will result in a strictly lower travel cost for all users in that group.}}
      \label{fig:prop-4-countereg}
   \end{figure}

Using the above scenario, we first derive the untolled user equilibrium solution, which allows to characterize the refunding scheme $\r$. Then we compute an exogenous equilibrium and finally show that this is not an endogenous equilibrium by describing a deviation which improves the travel costs of users in group~$M$. 

\paragraph{Untolled User Equilibrium.} Under this setting it is easy to see that at the user equilibrium the flow on the two edges are $x_1^{UE} = x_2^{UE} = 4$ giving a travel time of $8$ to all users. Now, under the UE solution, the costs to the three groups are:
\begin{align*}
    \text{Group H User Travel Cost} &= 8(0.002q_M) = 0.016q_M, \\
    \text{Group M User Travel Cost} &= 0.008q_M, \\
    \text{Group L User Travel Cost} &= 0.008q_L.
\end{align*}

\paragraph{Exogenous Equilibrium Under Tolls.} Now consider a setting with tolls where $\tau_1 = 0.008q_M$. In this setting, the traffic equilibrium (when refunds are exogenous) is given by $x_1^{\tau_1} = 2$ and $x_2^{\tau_1} = 6$, since only users in group $H$ are willing to travel on edge $1$ at this toll. This leads to a travel time of $t_1(x_1^{\tau_1}) = 4$ on edge 1, and a travel time of $t_2(x_2^{\tau_1}) = 10$ on edge 2. Note that under this equilibrium, the costs to the different groups are:
\begin{align*}
    \text{Group H User Travel Cost without Refunds} &= 4(0.002q_M)+0.008q_M = 0.016q_M, \\
    \text{Group M User Travel Cost without Refunds} &= 10(0.001q_M) = 0.010q_M, \\
    \text{Group L User Travel Cost without Refunds} &= 0.010q_L.
\end{align*}
Note that the total toll revenues collected are $x_1^{\tau_1} \tau_1 = 2(0.008q_M) = 0.016q_M$. Based on the refunding policy, we must first ensure that we give enough refunds so that each user incurs the same cost as that under the untolled user equilibrium. Then any remaining refunds are given to users based on the policy given in Algorithm~\ref{alg:greedy}. Following this procedure of giving refunds, we obtain the following set of refunds for each user group:
\begin{align*}
    \text{Total Refund given to Group H} &= 0, \\
    \text{Total Refund given to Group M} &= 0.002 q_M, \\
    \text{Total Refund given to Group L} &= 0.014q_M.
\end{align*}
Note here that users in group $H$ are given no refund since their cost under the tolled exogenous equilibrium is the same as that under the untolled user equilibrium. Group $M$ is given a refund that is exactly equal to the difference travel cost of the user group under the tolled exogenous equilibrium and the untolled user equilibrium, i.e., $0.010q_M - 0.008q_M = 0.002q_M$. Finally, the remaining refund of $0.016q_M - 0.002q_M = 0.014q_M$ is given to users in group $L$. This is because each user in group $L$ receives a refund of $\frac{0.014q^M}{5}$ and at this amount of refund, users in group $M$ and $L$ have exactly the same income since $q_M (1-0.008) = q_L (1-0.010) + \frac{0.014q^M}{5}$.

\paragraph{Profitable Deviation under Endogenous Equilibrium with Coalitions.} Now, we claim that users in group $M$ have a profitable deviation under the above refunding policy. In particular, we can consider the deviation where all of the users in group $M$ deviate to link 1. In this case, we have the edge flows $\Tilde{x}_1^{\tau_1} = 3$ and $\Tilde{x}_2^{\tau_1} = 5$, and so the travel time on edge 1 is $t_1(x_1^{\tau_1}) = 6$ and that on edge 2 is $t_2(x_2^{\tau_1}) = 9$. This leads to the following total cost to the different groups without refund:
\begin{align*}
    \text{Group H User Travel Cost without Refunds} &= 6(0.002q_M)+0.008q_M = 0.020q_M, \\
    \text{Group M User Travel Cost without Refunds} &= 6(0.001q_M)+0.008q_M = 0.014q_M, \\
    \text{Group L User Travel Cost without Refunds} &= 9(0.001q_L) = 0.009q_L.
\end{align*}
Then we have that the total tolls collected are $\Tilde{x}_1^{\tau_1} \tau_1 = 3(0.008q_M) = 0.024q_M$. Now, we must transfer the amount $0.004q^M$ to group $H$ to equalize their cost to the untolled UE cost. For group $M$, we must transfer $0.006q^M$ to equalize their cost to the untolled UE cost. This leaves a total refund of $0.014q^M$. Now, we will distribute this remaining income such that the incomes of users in group $M$ and $L$ are equal. Note that in doing so we will also have ensured that users in group $L$ are at least as well off as under the untolled user equilibrium outcome. In particular, we look to find the value $y$ that satisfies the following equation:
\begin{align*}
    \frac{y}{5} q_M + q_L (1-0.009) = (0.014-y) q_M + q_M(1-0.008).
\end{align*}
Here $y$ is some constant. We can solve for $y$ using the above relation that $q_M (1-0.008) = q_L (1-0.010) + \frac{0.014q_M}{5}$, which gives us that $0.014-y>0$, implying that users in group $M$ get strictly lower costs than the untolled user equilibrium outcome. Since the cost of user group $M$ was exactly their untolled user equilibrium cost, we observe here that user group $M$ has a profitable deviation implying that the exogenous equilibrium is not an equilibrium when coalitions of users take revenue refunds into account in their travel decisions.
\fi

\section{Numerical Experiments} \label{apdx:numerical_main}

In this section, we present numerical experiments to demonstrate the efficacy of optimal CPRR schemes in reducing the total system cost and wealth inequality and also show that the benefits of CPRR can even be realized in the setting when users' values of time are not known to the central planner. To this end, we first discuss the  implementation details of our experiments as well as the chosen network structures, O-D demands, travel-time functions, user values of time, and incomes (Section~\ref{sec:exp-details}). Then, we present the results of our experiments in Section~\ref{sec:results-experiments}.



\subsection{Experimental Details} \label{sec:exp-details}

\begin{figure}
    \centering
    \includegraphics[width=0.8\linewidth]{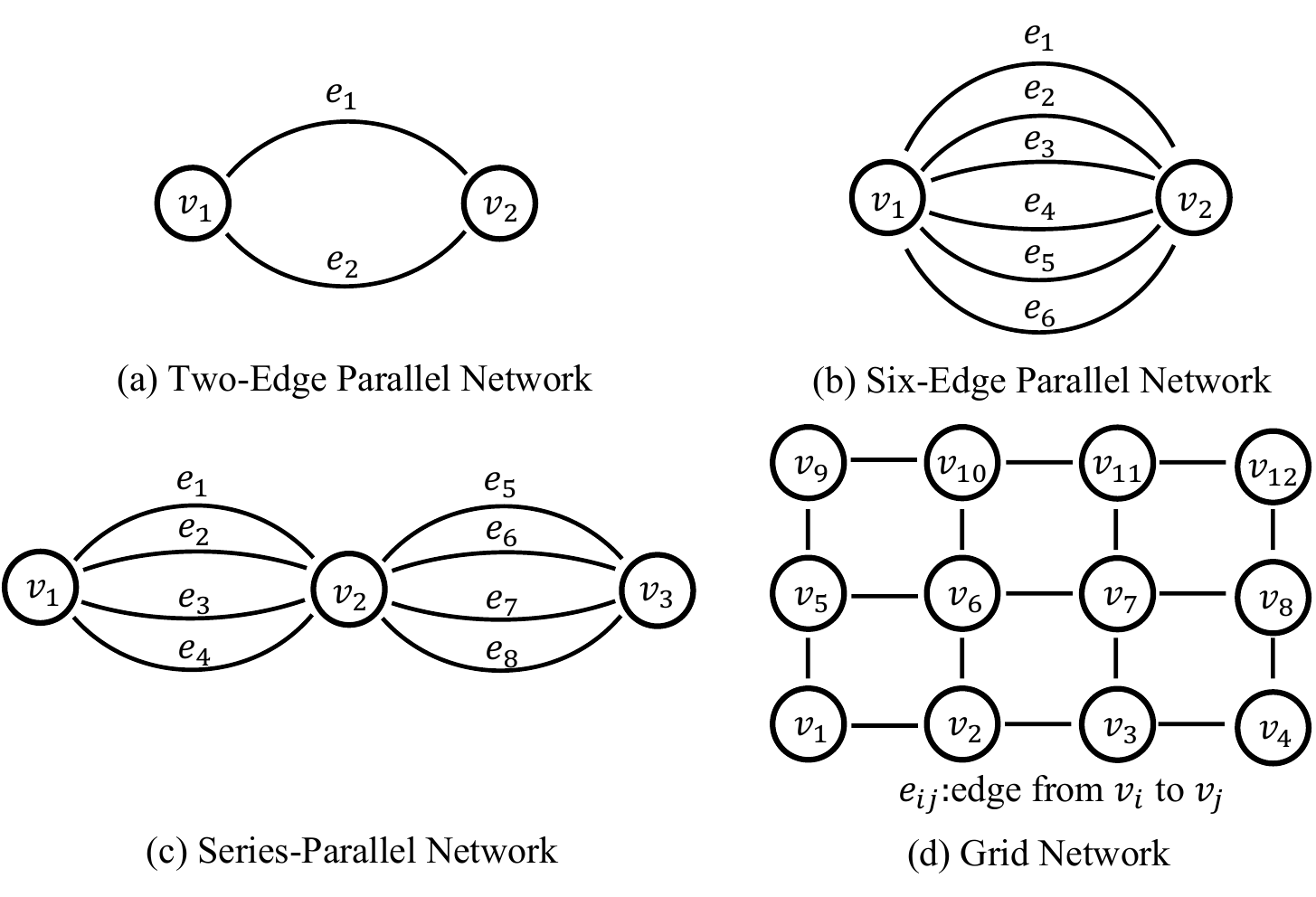}
    \caption{{\small \sf Depiction of road networks considered for the numerical study}}
    \label{fig:networks}
\end{figure}

\textbf{Road networks:} The four networks used in our experiments are shown in Figure \ref{fig:networks}. We assume that the travel time function on each edge increases linearly with the flow and 
the numerical values of the parameters (i.e., intercept and slope) of the travel time functions for all the edges in the different networks in Figure \ref{fig:networks} is available at \href{https://github.com/karthikg92/cprr-tcns}{https://github.com/karthikg92/cprr-tcns}. 

\noindent
\textbf{Demand Scenarios:} 
For the Pigou, Parallel and Series-parallel networks, we consider an O-D demand between one OD pair ($v_1$ to $v_2$ for the Pigou and Parallel network, and from $v_1$ to $v_3$ in the Series-Parallel network) and for the grid network consider two O-D pairs ($v_1$ to $v_7$ and $v_3$ to $v_8$). For the Pigou, Parallel and Series-parallel networks, we considered six user groups with varying income levels and values of time, and for the grid network, we considered four user groups with varying income levels and values of time. We note that for our experiments for the  Pigou, Parallel and Series-parallel networks, we assumed that the values of time are chosen proportional to user incomes, a commonly used modelling assumption~\cite{orpn-2023}. 
The demand for the grid network is designed to simulate a more realistic scenario where the value of time may not be proportional to the income of the users. To this end, we consider three different scenarios to represent the degree of variability (denoted by Low Var, Med Var, High Var) in the values of time of users. The numerical values for the O-D demands, user incomes, values of time of users in the different scenarios are all available at \href{https://github.com/karthikg92/cprr-tcns}{https://github.com/karthikg92/cprr-tcns}.

\noindent
\textbf{Approximate CPRR Computation:} 
In addition to evaluating the performance of optimal CPRR schemes, we also perform experiments in the incomplete information setting when user specific values of time or incomes may not be known, as is often the case in practice. In this incomplete information setting, we only assume access to the mean values of time and incomes of users and provide all users travelling between a given O-D pair the same refund, i.e., we consider anonymous refunding schemes as in~\cite{GUO2010972}. To compute CPRR schemes in this incomplete information setting, we aggregate users with the same O-D pair into one group and use their mean value of time to (1) solve the user equilibrium problem with zero tolls, (2) estimate the optimal tolls by solving the system optimal problem, and (3) compute the refunds.
However, we note that we compute user responses to the computed approximate tolls using exact value of time information of the users. Note that under our approximation, all users having the same O-D pair receive the same refunds irrespective of their values of time or incomes.

\subsection{Results} \label{sec:results-experiments}

We now present the results of the comparison between the total system cost and the level of wealth inequality of the ex-post income distribution for both the optimal CPRR scheme and that computed under incomplete information. In particular, Table~\ref{tab:simulation} presents the relative percentage differences of the total system cost and wealth inequality of the ex-post income distribution for the optimal CPRR scheme and the one under incomplete information to the user equilibrium outcome without tolls and refunds.

We first note from Table~\ref{tab:simulation} that the optimal CPRR scheme, as expected, reduces the total system cost and the wealth inequality measure, i.e., the discrete Gini coefficient, compared to the user equilibrium setting with no tolls or refunds, thereby corroborating Proposition~\ref{prop:rev-refund-decreases-ineq}. In addition, since users' values of time are assumed to be scaled proportions of their incomes for the purposes of the experiments~\cite{orpn-2023}, our results for the optimal CPRR scheme for single O-D pair demand, i.e., for the ``Pigou'', ``Parallel'' and ``Series-Parallel'' networks, also corroborate Corollary~\ref{cor:rev-refund-single-od}, as the discrete Gini coefficient of the ex-post income distribution under the optimal CPRR scheme is lower than that corresponding to the ex-ante income distribution (see column 5 in Table~\ref{tab:simulation}). 

We also perform experiments in the incomplete information setting when user specific values of time or incomes may not be known, as is often the case in practice. In this incomplete information setting, we reiterate that we only assume access to the mean values of time and incomes of users and provide all users travelling between a given O-D pair the same refund, i.e., we consider anonymous refunding schemes as in~\cite{GUO2010972}. Our results in Table~\ref{tab:simulation} indicate that deploying CPRR schemes in this setting generally results in total system costs and level of wealth inequality that are higher than that of the optimal CPRR schemes in the complete information setting but lower than that corresponding to the user equilibrium setting with no tolls and refunds. Table~\ref{tab:simulation} also indicates that the performance of the CPRR scheme with incomplete information depends on the variance in the user values of time and income around the mean. In particular, Table~\ref{tab:simulation} indicates that as the variance in user values of time is decreased, the CPRR scheme with incomplete information achieves a performance closer to that of the optimal CPRR scheme on both total system cost and wealth inequality metrics.

\begin{table}[] 
\centering
\caption{{\small \sf Relative percentage differences of the total system cost (columns 1 and 2) and wealth inequality (columns 3 and 4), evaluated by the discrete Gini coefficient, of the optimal CPRR scheme (columns 1 and 3) and that with incomplete information (columns 2 and 4) compared to the user equilibrium outcome without tolls on four traffic networks: (i) Pigou network, (ii) Four edge parallel network, (iii) Series Parallel Network, and (iv) Grid network. For the grid network, two O-D pairs were considered for three settings depending on the degree of variance of users' values of time, i.e., low, medium, or high. Here, $C_I$ and $\q^I$ denote the total system cost and ex-post income distribution for the scheme with incomplete information, $W^* = W(\q(\ttau^*, \r^*))$, and $W^0 = W(\q(\0, \0))$. Column 5 represents the relative percentage difference of the wealth inequality of the optimal CPRR scheme compared to the ex-ante income distribution.}}
\begin{tabular}{c|c|c|c|c|c}
\toprule
Experiment      & $\frac{C_{\0} - C^*}{C_{\0}}$  & $\frac{C_{\0} - C_I}{C_{\0}}$  & $\frac{W^0-W^*}{W^0}$ & $\frac{W^0-W(\q^I)}{W^0}$ & $\frac{W(\q^0)-W^*}{W(\q^0)}$ \\
\midrule
Pigou (2 edge)         &  5.1147             & 5.1029      & 0.0357               & 0.0297      &    0.0262       \\
Parallel (4 edge)       &   4.1343            & 4.1223      & 0.0167               & 0.0134      &   0.0113        \\
Series-Parallel &  4.8809             & 4.8331      & 0.0609               & 0.0554          &   0.2437    \\
Grid (Low Var)  &  0.9910             & 0.9834      & 0.0107               & 0.0071         &    -0.1651    \\
Grid (Med Var)  &  1.4824             & 1.3062      & 0.0161               & 0.0070          &    -0.1889   \\
Grid (High Var) &  2.3365             & 1.6787      & 0.0253               & 0.0070           &   -0.2160  
\end{tabular} \label{tab:simulation}
\end{table}

\end{document}